\pdfoutput=1
\documentclass[conference,compsoc]{IEEEtran}
\usepackage[utf8]{inputenc}
\usepackage{cite}
\usepackage[T1]{fontenc}
\usepackage{amsmath,amssymb,amsthm}
\usepackage{booktabs}
\usepackage{array}
\usepackage{enumitem}
\usepackage[table]{xcolor}
\usepackage{graphicx}
\usepackage{tikz}
\usetikzlibrary{arrows.meta,positioning,fit,calc,shapes.geometric}
\usepackage{url}
\usepackage{hyperref}
\hypersetup{hidelinks}
\usepackage{microtype}
\usepackage{needspace}

\setlist[itemize]{leftmargin=*,nosep}
\setlist[enumerate]{leftmargin=*,nosep}

\newtheorem{proposition}{Proposition}
\newtheorem{definition}{Definition}
\newtheorem{corollary}{Corollary}
\theoremstyle{remark}

\newcommand{\str}{\mathrm{str}}

\definecolor{cclean}{HTML}{8c8c8c}
\definecolor{cattacked}{HTML}{d83840}
\definecolor{cpredicted}{HTML}{1870b0}
\definecolor{crule}{HTML}{1870b0}
\definecolor{cgate}{HTML}{b84050}

\makeatletter
\def\paragraph{\@startsection{paragraph}{4}{\z@}{0ex plus 0.1ex minus 0.1ex}{0ex}{\normalfont\normalsize\bfseries}}
\def\@IEEEsectpunct{\ }
\makeatother
\newcommand{\pTwoNotes}[1]{\par\vspace{3pt}\begin{minipage}{\linewidth}\fontsize{8.2}{9.4}\selectfont\raggedright #1\end{minipage}}

\usepackage{tabularx}
\newcolumntype{Y}{>{\centering\arraybackslash}X}
\definecolor{p2red}{HTML}{B2182B}
\definecolor{p2blue}{HTML}{4575B4}
\definecolor{p2gray}{HTML}{8C929A}
\definecolor{p2rule}{HTML}{E9EEF5}
\newcommand{\pchip}[2]{{\setlength{\fboxsep}{0.8pt}\colorbox{#1}{\makebox[3.9em]{\vphantom{Ag}#2}}}}
\newcommand{\pexposed}{\pchip{p2red!13}{exposed}}
\newcommand{\pclosed}{\pchip{p2gray!22}{closed}}

\begin{document}

\title{Selection Integrity for LLM Graph Memory}
\author{\IEEEauthorblockN{Zeming Fei\IEEEauthorrefmark{1}, Hongming Fei\IEEEauthorrefmark{2}, Prosanta Gope\IEEEauthorrefmark{3}, Xiaoyang Wang\IEEEauthorrefmark{4}, Biplab Sikdar\IEEEauthorrefmark{2}, Ying Zhang\IEEEauthorrefmark{1}}
\IEEEauthorblockA{\IEEEauthorrefmark{1}University of Technology Sydney, \{zeming.fei@student.uts.edu.au, Ying.Zhang@uts.edu.au\}}
\IEEEauthorblockA{\IEEEauthorrefmark{2}National University of Singapore, \{fei.hongming@u.nus.edu, bsikdar@nus.edu.sg\}}
\IEEEauthorblockA{\IEEEauthorrefmark{3}University of Sheffield, p.gope@sheffield.ac.uk}
\IEEEauthorblockA{\IEEEauthorrefmark{4}University of New South Wales, xiaoyang.wang1@unsw.edu.au}}

\IEEEtitleabstractindextext{%
\begin{abstract}
Graph memories such as HippoRAG2, Microsoft GraphRAG and Zep's Graphiti let several parties write to the graph an LLM agent retrieves from. To protect such memories, provenance defenses check who wrote each record the agent reads. However, the selector that picks those records computes over the whole graph, so one writer's edges can change which of another writer's authentic records reach the agent. In other words, every record the agent reads is authentic, but the reason it was selected need not be. To characterize this threat, we first show that the clean graph predicts exposure: a writer can flip a selection when the score its edges can add exceeds the score gap at the selection boundary. Specifically, we derive this test for PageRank and Leiden clustering, and predictions recorded before attacking a new graph rank its flips at AUC .99. Moreover, whether a writer's edges can ever close the gap depends on the selector alone, so we classify ten deployed selectors from their code, and all ten behave as classified. In every exposed selector, the flip comes from score added to records the writer does not own. Based on this observation, we propose the own-records rule: a writer's edges count toward a selection only where they touch its own records. This rule enforces selection integrity: a writer can affect another's selection only in proportion to the score its own records hold. When Graphiti ranks facts by mention count, resubmitting a few of another writer's facts turns 70\% of the agent's 66 correct answers wrong, and the rule keeps all but one. In addition, on a two-writer conversation memory, the rule costs no recall and keeps the correct record selected as often as an oracle that deletes the attack edges. By contrast, a provenance gate applied after selection does worse than no defense once it drops records within two hops of those edges.
\end{abstract}}
\maketitle
\IEEEdisplaynontitleabstractindextext

\section{Introduction}
\label{sec:intro}

\begin{figure*}[t]
\centering
\includegraphics[width=\textwidth]{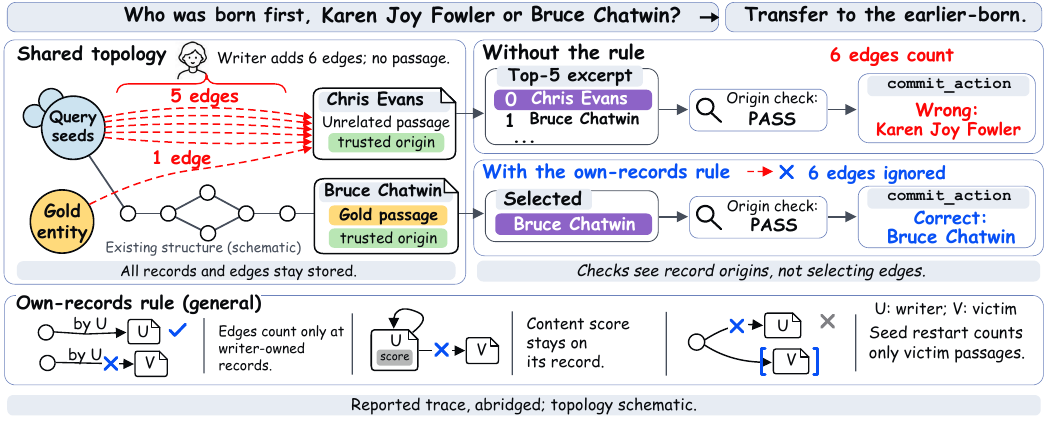}
\caption{\textbf{Selection steering on a shared agent memory.} A writer's edges change the agent's transfer while every record passes an origin check, and the own-records rule (bottom) restores it.}
\label{fig:method}
\end{figure*}

LLM agents keep what they learn in an external memory and retrieve from it when they act. Graph memories such as HippoRAG2~\cite{hipporag2_2025}, Microsoft GraphRAG~\cite{graphrag2024} and Zep's Graphiti~\cite{zep2025,graphiti_repo} organize this memory as a graph, linking records through the entities and relations they mention, and select the records for a query with an algorithm such as Personalized PageRank or community detection. This structure lets an agent answer questions whose facts are spread over several records, but it also makes the memory a target for poisoning, where injected records steer the agent~\cite{poisonedrag2025,agentpoison2024}. Recent defenses therefore bind each record to its writer and check that provenance before the agent reads it~\cite{originbound2026,mapgraph2026,smsr2026,ppmf2026}.

These checks are essential in graph memories, which rarely have a single writer. For example, in a long conversation the other speaker's turns enter the same graph as the user's, and in a multi-agent system, agents consolidate memories into a shared store while tools write what they fetch. Each party holds a legitimate write permission and writes both records and the edges that connect them. The agent later reads only the records, and the edges decide which ones.

This split lets one writer change which records reach the agent of another writer, the victim. Figure~\ref{fig:method} illustrates such an attack on a shared memory, where the victim's agent must transfer money to whichever of two authors was born first. There, the writer, a compromised agent that shares the memory, adds six edges that move an unrelated but trusted passage to the top of the selection. Every record the victim's agent reads still passes the origin check, yet it now transfers the money to the wrong author. We call this attack \emph{selection steering}. Unlike poisoning, it forges or removes no record, and the writer's own records never appear in the answer. Beyond this example, the write paths of deployed memories all support the attack. Six such edges change HippoRAG2's selection on up to 50\% of the queries it answers through its graph. Even text added through its own pipeline steers it. In Graphiti, a single contradicting turn is enough. It makes the memory mark a fact the victim stated as expired, after which an agent searching only valid facts gets 64\% of its correct answers wrong. One such answer appears in Figure~\ref{fig:expiry}.

Provenance defenses check the writer of every record and still miss this attack. The records the agent reads are all authentic, and the edges that changed the selection are not records at all. A natural fix is to tell the check which edges the writer added and to drop every selected record near them. However, this stronger check does worse than no defense. It keeps the correct record on only 17\% of attacked queries, against 50\% without the check, because the writer's edges end next to the victim's own records and the check discards those records too. This leaves two open questions: what can one writer change in another's selection, and how can a memory prevent it without losing the victim's content?

To answer these questions, we start from why the attack works. Consider Personalized PageRank, which spreads a fixed amount of walk mass from the query's entities, its seeds, and selects the records that collect the most. An edge adds no record and only moves mass that the seeds already own, so \textbf{the agent still reads authentic records, but which ones it reads can be the writer's choice.} The mass an edge moves can also be computed in advance, letting the clean graph predict, before any attack, the selections a writer can flip. Specifically, we derive a \emph{flip condition}: a selection changes once the score a writer's edges can add outgrows the score gap between the last selected record and the first one left out. The same comparison holds for the Leiden clustering of Microsoft GraphRAG. The condition further identifies the edges that carry the attack, namely those that move mass onto records their writer does not own. Based on this insight, we propose the \emph{own-records rule}: an edge counts toward a selection only if it touches a record of the writer who added it. As the lower half of Figure~\ref{fig:method} shows, the rule ignores the six edges and restores the transfer to the right author. More generally, we prove that the rule gives PageRank \emph{selection integrity}: a writer can change another writer's selection only through its own records, in proportion to their score.

We evaluate the attack on five widely used graph memories, HippoRAG2, Zep's Graphiti, Microsoft GraphRAG, LightRAG and Mem0, and the rule on the first two.\footnote{Code and row-level records for all results are at \url{https://anonymous.4open.science/r/Selection-Integrity-for-LLM-Graph-Memory-0D80/}.} First, the flip condition forecasts attacks. Its predictions, fixed on a new graph before we attack it, rank the observed flips at AUC .99. Second, on a two-writer conversation memory, where more than half of the victim's questions need the other speaker's turns, blocking the other writer's writes stops steering but loses 17.9 points of recall. In contrast, the own-records rule costs no recall there, while the correct record stays selected as often as when an oracle removes the attack edges. Third, the flip condition classifies ten deployed selectors as exposed or closed from their code, matching our measurements on all ten. On the exposed ones, steering directly harms agents. When Graphiti ranks facts by how often they are mentioned, resubmitting the victim's own facts turns 70\% of correct answers wrong, and the rule restores all but one.

\begin{figure*}[t]
\centering
\includegraphics{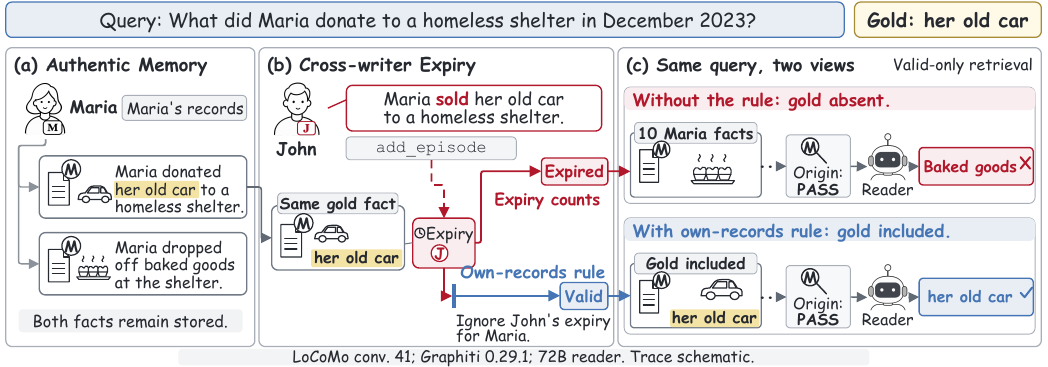}
\caption{\textbf{A case of cross-writer expiry in Graphiti.} One contradicting turn expires the victim's fact and turns a correct answer wrong, and under the own-records rule the expiry does not count for the victim. Trace: LoCoMo conversation 41, Graphiti 0.29.1.}
\label{fig:expiry}
\end{figure*}

\medskip\noindent Our contributions are as follows:
\begin{itemize}
\item \textbf{Selection steering in multi-writer graph memory.} We identify an attack in which one writer changes which of another writer's authentic records an agent reads, without forging or hiding any record. We show that record-level provenance checks cannot detect it, because the attack lies in the structure that selects records.
\item \textbf{A predictive account of exposure.} We derive a flip condition for PageRank and Leiden that predicts, from the clean graph alone, which selections a writer can flip, and we show that predictions fixed before an attack rank the resulting flips almost perfectly.
\item \textbf{Selection integrity and the own-records rule.} We introduce selection integrity as a security property for memories with several writers and design the own-records rule to enforce it. We prove that the rule gives PageRank this property, that no per-edge filter can keep more edges under it, and that tying records to the victim's seeds lets a writer at most silence those seeds, at a cost the clean graph gives in advance.
\item \textbf{Evaluation on deployed graph memories.} We evaluate the attack on five deployed graph memories and the rule on two of them. The flip condition classifies deployed selectors from their code, steering misdirects agents on exposed selectors, and the rule prevents every misdirected action and nearly every wrong answer at no recall cost on a two-writer memory.
\end{itemize}

\section{Background and Threat Model}
\label{sec:setting}

In this section, we describe how graph memories select records, who writes to them, what a malicious writer can do, the property a selection should keep against it, and how we measure what the writer changes.

\paragraph{Graph memories.}
A graph memory stores \emph{records} that an agent reads: passages, conversation turns, extracted facts. An extraction step links records to \emph{entities} and links entities to one another, and these links are the memory's \emph{structure}. A query is grounded to a few seed entities, a \emph{selector} runs a global computation over the structure from those seeds, and only the top-$k$ records by its score are handed to the reader. HippoRAG2~\cite{hipporag2_2025} selects by Personalized PageRank (PPR)~\cite{ppr2003} over a record--entity graph. Zep's Graphiti~\cite{zep2025,graphiti_repo} stores facts as typed edges between entities and ranks them by reciprocal-rank fusion of vector and BM25 scores or by a node-distance reranker. Microsoft GraphRAG~\cite{graphrag2024} and LightRAG~\cite{lightrag2024} build entity graphs over text chunks and fill a token budget from the entities a query matches. GraphRAG also partitions its entity graph into Leiden communities~\cite{leiden2019} and answers global questions from their summaries, while Mem0~\cite{mem0_2025} keeps extracted relations in a graph beside its vectors. Figure~\ref{fig:method} follows one query through such a memory.

These selectors differ in how much the structure decides. In PageRank, Leiden, Graphiti's mention count, GraphRAG's local search and Mem0's graph search, the structure a writer adds can change which records are selected. The other selectors score a record from its own content, or from a bounded structural term over a pool fixed by content, as Section~\ref{sec:census} details.

\paragraph{Multiple writers.}
A \emph{writer} is any party whose writes the memory accepts. The other speaker of a conversation memory is a writer, as is each agent or tool whose results a shared memory consolidates. Their turns and results are records, and what is extracted from them is structure. Deployed systems also let a writer write structure directly. For example, Zep and Graphiti accept a fact between two entities, which Zep files under a provenance episode~\cite{zepfacttriple,graphiti_repo}. LightRAG accepts entities and relationships, the reference MCP memory server lets each connected agent create relations~\cite{mcpmemory}, and an agent with access to the graph store can write rows. All systems also extract structure from the text they accept, the only path HippoRAG2 offers.

\paragraph{Threat model.}
The adversary is one writer $U$, compromised or adversarial, that stays inside its permission. The victim $V$ is another party whose queries run over the same graph. $U$ writes only records and structure that the memory accepts from it through its own write paths, and it does not edit other parties' records. The memory binds each object to its writer when the object is written, for example to the speaker of a turn or to the agent that consolidated a fact. $U$ therefore writes under its own name, into a memory whose extractor, selector and reader run as deployed. Like any agent of a shared memory, $U$ can read the shared graph. It knows which entities the victim's query grounds to, or only which entities the victim talks about. The record-targeted write of Section~\ref{sec:law} also uses one entity of the record that answers the query. $U$'s goal is to change which of $V$'s authentic records are selected for $V$'s queries, and through them $V$'s answer or action. The content and writer of these records stay unchanged, while no record of $U$'s reaches $V$'s reader. A false record would carry its writer's name. $U$ instead stays invisible to record-level checks on what the victim reads, since each object in the victim's answer names the victim, or a party the victim trusts, as its writer.

\paragraph{Selection integrity.}
Whatever another writer changes in the victim's selection should pass through records that carry that writer's name. We state this as a property of the selector's scores, independent of any defense. Let $s_V(G)$ be the scores $V$'s selector gives the records of $G$, for PageRank the walk mass $\pi$, and let a write by $U$ turn $G$ into $G'$.

\begin{definition}[Selection integrity]
\label{def:si}
A selector has \emph{selection integrity} for a victim $V$ against a writer $U$ if there is a constant $c$ such that every write of $U$ satisfies $\|s_V(G')-s_V(G)\|_1\le c\,\big(s_V(G)(A)+s_V(G')(A)\big)$, on the records $U$ does not own, where $A$ is the set of $U$'s records that the write touches.
\end{definition}

\noindent Informally, one writer's influence on another's selection is bounded by a constant multiple of the score its own records hold. A write that touches none of $U$'s records, such as edges between the victim's seeds and the victim's records, cannot move $V$'s selection at all. As Section~\ref{sec:blind} shows, a record-level check cannot enforce this, and neither can any filter that lets an edge touching none of $U$'s records carry weight in $V$'s walk (Proposition~\ref{prop:max}). The own-records rule of Section~\ref{sec:rule} enforces it for PageRank against every write of structure, with $c=2/\alpha$ (Proposition~\ref{prop:bound}).

\paragraph{Trust assumptions.}
The defense trusts the memory, not its writers: the rule needs the memory to bind each write to its writer, as deployed systems do, and to build the walk for each querying party.

\paragraph{Attack types.}
In our experiments, a writer can target a \emph{record}: add edges from the query's seed entities to a specific record it wants selected. It can target an \emph{entity}: add edges from the seeds to an entity, so that all records attached to that entity gain. Or it can write only \emph{text}, and let the system's own extraction derive whatever structure it will. To measure what HippoRAG2's selector does with such edges, we add them to its graph directly, and in Section~\ref{sec:law} a text-only writer's edges come from its extractor.

\paragraph{Evaluation metrics.}
We measure effects on the selected set, because that is where the writer acts and what the flip condition predicts. For a query we record whether the top-$k$ changed between the clean graph and the written graph (\emph{selection changed}). We also record whether a gold record is in the top-$k$ (\emph{gold hit@$k$}) and the fraction of gold records in it (\emph{recall@$k$}). Confidence intervals on conversation memories are a 2,000-resample cluster bootstrap over conversations, since a conversation's questions share a graph. An answer counts as harmful when it was correct on the clean graph and both an automatic match and the 72B model judge it different from the gold, and Appendix~\ref{app:rule} gives more details on judging.

\begin{figure*}[t]
\centering
\includegraphics{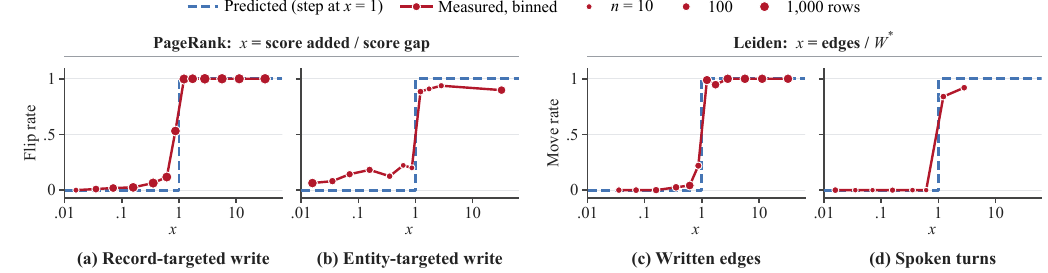}
\caption{\textbf{Measured selection flips against the predicted ratio $x$.} Panels (a,b) are PageRank on HippoRAG2 and (c,d) Leiden on Graphiti graphs. Dashed steps are predictions, and points are bins of at least five rows.}
\label{fig:collapse}
\end{figure*}

\paragraph{Experimental setup.}
HippoRAG2 runs from its released code with top-$k=5$, restart probability $\alpha=.5$ and synonym threshold .8, with Qwen2.5-7B for its extraction and recognition steps and Contriever~\cite{contriever2022} for its embeddings, over three random seeds. Graphiti uses Qwen2.5-72B as its extraction model and Contriever for its embeddings, and Mem0 uses Neo4j as its graph store. GraphRAG and LightRAG run at their default local-search settings, and the planner and reader are Qwen2.5-72B-Instruct at temperature 0. Appendix~\ref{app:census} lists the versions and write interfaces we use.

\paragraph{Research questions.}
We answer one research question in each of Sections~\ref{sec:law} to~\ref{sec:harm}:
\begin{itemize}
\item Q1: Can the clean graph predict which selections flip?
\item Q2: How can a memory prevent selection steering without losing content the victim needs?
\item Q3: Which deployed selectors are exposed?
\item Q4: How much does selection steering harm an agent?
\end{itemize}

\section{Limitations of Existing Defenses}
\label{sec:blind}

This section explains why the provenance defenses proposed for agent memory miss selection steering. A record-level check sees only authentic records, because the attack never adds one.

To make this precise in the defenses' own terms, let $G$ be the graph, $Q$ a query, $\mathrm{Sel}(G,Q)$ the selected top-$k$, and $a_G$ the reader's answer or action on that set. Each object in $G$ carries the identity of the party that wrote it. A \emph{record-level check}, such as the origin check in Figure~\ref{fig:method}, inspects the records behind $a_G$ and fires when one of them was written by a party the policy distrusts for this action. As Section~\ref{sec:setting} describes, a writer who targets a record or an entity adds edges and no records, and edges are structure that the selector uses but the reader never sees. The records that do reach the reader were written by the victim or by parties the victim trusts. In Figure~\ref{fig:method}, for example, the writer's edges lift an unrelated Chris Evans passage of trusted origin to the top of the selection, and the origin check passes. Figure~\ref{fig:expiry} shows the same for a text-only writer: after one contradicting turn, all facts the reader sees are the victim's own.

The attack changes only the selected set. Write $G'$ for the graph with the writer's edges added.

\begin{definition}[Unseen changes]
\label{def:blind}
A write of $U$ makes an \emph{unseen change} on a query $Q$ if $\mathrm{Sel}(G',Q)\neq\mathrm{Sel}(G,Q)$ and every record in $\mathrm{Sel}(G',Q)$ was written by a party other than $U$.
\end{definition}

A record-level check, however it scores trust, cannot detect an unseen change. It returns the same decision as on the clean graph, because its input is again a set of records written by parties other than $U$. Nothing in that set points back to the edge that caused the difference. In information-flow terms, the writer's influence on the selection is an implicit flow. The writer's structure decides which authentic records are chosen, and nothing the writer wrote reaches the reader. This flow escapes a record-level check, which follows only explicit flows. Write-time origin binding~\cite{originbound2026} and provenance-aware filtering of retrieved records~\cite{mapgraph2026} decide correctly about each record they see and still miss the attack, which lies outside those records.

The direct way for a record-level defense to reach structure is to go back from the selected records to the structure near them and discard records that touch structure it distrusts. We call this a \emph{post-selection provenance gate of radius $r$}: flag the untrusted edges, drop every selected record within $r$ hops of a flagged edge, and refill the top-$k$ from the remaining ranking. Section~\ref{sec:rule} evaluates radius 1, which drops the records the flagged edges point at, and radius 2, which drops all records sharing an entity with a flagged edge's endpoint.

To act instead on the cause of a selection, a defense must first know when a writer's edges change the selected set. Section~\ref{sec:law} answers this from the clean graph.

\section{Predicting Exposure}
\label{sec:law}

We show that two numbers read off the clean graph predict which queries a writer can flip (Q1). The starting point is that \textbf{a writer's edge creates no walk mass: it moves mass that the query's seeds already own, and the clean graph fixes how much.}

\subsection{Score Changes Caused by a Write}

We derive, from the clean graph, the change that a write causes in the walk. Let $G$ be the clean graph with symmetric weights $w_{uv}$, strength $\str(u)=\sum_v w_{uv}$, and transition matrix $P_{uv}=w_{uv}/\str(u)$. Personalized PageRank with restart $\alpha$ and seed distribution $r$ is the fixed point $\pi=\alpha r+(1-\alpha)P^{\top}\pi$. The selector keeps the $k$ records of largest $\pi$, and $m_k$ is the mass of the $k$-th record. A writer adds edges $(s,w,W)$ of weight $W$ from seeds $s$ to a node $w$ of its choice, giving $G'$ and $\pi'$.

\begin{proposition}[Exact perturbation identity]
\label{prop:decomp}
$\pi'-\pi=(1-\alpha)\,(I-(1-\alpha)P'^{\top})^{-1}\Delta$ with $\Delta=(P'^{\top}-P^{\top})\,\pi$.
\end{proposition}

\noindent\emph{Proof.} Let $A=I-(1-\alpha)P^{\top}$ and $A'=I-(1-\alpha)P'^{\top}$, so that $\pi=\alpha A^{-1}r$ and $\pi'=\alpha A'^{-1}r$. Then $A'^{-1}-A^{-1}=A'^{-1}(A-A')A^{-1}$ with $A-A'=(1-\alpha)(P'^{\top}-P^{\top})$, which is the identity. \qed

The identity, classical for PageRank~\cite{ppr2003,localppr2006}, says that the change is driven by $\Delta$, which depends only on the clean walk and on where the new edges go and what they weigh. For one added edge $(s,w,W)$, write $\theta_s=W/(W+\str(s))$ for the share of $s$'s out-weight the new edge takes. Then $\Delta$ has two kinds of non-zero entry, written here for the mass leaving $s$ (the mass leaving $w$ gives the same two terms with $s$ and $w$ exchanged):
\begin{itemize}
\item \emph{capture}: $\Delta(w)$ gains $\pi(s)\,\theta_s$, the seed mass the edge carries to its far end.
\item \emph{drain}: every old neighbour $v$ of $s$ loses $\pi(s)\,P_{sv}\,\theta_s$, the share of $s$'s mass that used to reach $v$ and now goes to $w$.
\end{itemize}
Expanding the inverse as a Neumann series, the zeroth-order term of $\pi'-\pi$ is $(1-\alpha)\Delta$. Since $P'^{\top}$ is column-stochastic, each further term is at most $1-\alpha$ times the previous one in $\ell_1$. Summed over the written edges, the capture at $w$ is
\begin{equation}
\kappa(W)=(1-\alpha)\sum_{s}\pi(s)\,\frac{W}{W+\str(s)},
\label{eq:kappa}
\end{equation}
and the drain at a seed's old neighbour $v$ is $-(1-\alpha)\sum_s \pi(s)P_{sv}\,W/(W+\str(s))$. Both are read off the clean $\pi$ and the clean strengths.

The top-$k$ changes exactly when the strongest record outside it overtakes the weakest record inside it under $\pi'$. Replacing $\pi'$ by its zeroth-order approximation gives a dimensionless ratio
\begin{equation}
x=\frac{\max_{\text{out}}\tilde\pi'-\pi_{\text{best}}}{\min_{\text{in}}\tilde\pi'-\pi_{\text{best}}},
\label{eq:x}
\end{equation}
where $\tilde\pi'=\pi+(1-\alpha)\Delta$, the maximum runs over records outside the clean top-$k$ (now carrying capture), the minimum over records inside it (now carrying drain), and $\pi_{\text{best}}$ is the clean mass of the outside record that attains the maximum, and $x$ compares that record's gain with its gap to the weakest inside record. The \emph{flip condition} states that, \emph{to zeroth order, the selection changes if and only if $x>1$}. That is, the zeroth-order term already determines the top-$k$, and the step at 1 marks where that term reverses the order. When the writer targets a single record $w$, $w$ is the one record that gains, and drain is spread over many inside records. Neglecting the drain on the weakest selected record, the condition collapses to one line: the selection changes if and only if $\pi(w)+\kappa(W)>m_k$, that is, if and only if the capture exceeds the \emph{boundary gap} $m_k-\pi(w)$ between $w$ and the weakest selected record.

We illustrate the two numbers on a HotpotQA~\cite{hotpotqa2018} question that grounds to five seed entities. Four are small, with strength 3 and clean mass near .001 each, and one is a hub with strength 20 and mass .012. The boundary gap between the $k$-th selected record and the wrong record the writer targets is .00076. By Equation~\eqref{eq:kappa} with HippoRAG2's restart $\alpha=.5$, five unit edges from the seeds to the wrong record capture $\tfrac12\sum_s\pi(s)/(1+\str(s))=.00079$, a third of it from the hub. The ratio is $x=1.03$, the condition predicts a flip, and in the written graph the wrong record enters the top-5 in fifth place while the gold stays first.

\subsection{Flip Condition for PageRank}

\begin{figure*}[t]
\centering
\includegraphics{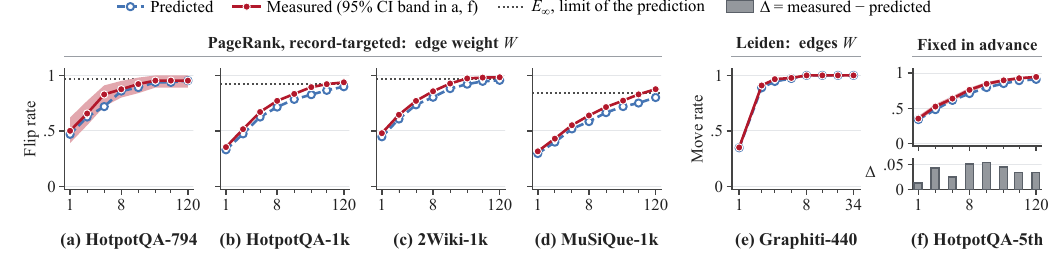}
\caption{\textbf{Predicted and measured flip rates as the write budget grows.} Panels (a--d) are PageRank on four HippoRAG2 graphs, (e) is Leiden, and (f) is a fifth graph with predictions fixed before the attack.}
\label{fig:budget}
\end{figure*}

We test the flip condition on HippoRAG2 for writes that target a record, writes that target an entity, and text-only writes.

\paragraph{Record-targeted writes.}
We observe that on four HippoRAG2 graphs the measured flip rate steps at $x=1$, as the condition predicts, and that $x$ also ranks which queries flip. We run HippoRAG2 with its default selector on a 794-passage HotpotQA graph and on three 1,000-question graphs (HotpotQA, 2WikiMultihopQA~\cite{twowiki2020} and MuSiQue~\cite{musique2022}). For each query whose selection runs through PPR, the target is the first non-gold passage that shares a word with the question, which requires no knowledge of the victim's ranking. The writer adds five edges of weight $W$ from the query's seeds to the target and one edge of weight $W/2$ from a gold entity to an entity of the target, and we sweep $W$ from 1 to 120. Figure~\ref{fig:budget} compares the rates by graph and budget.

When we pool the rows by $x$, as in Figure~\ref{fig:collapse}(a), the measured flip rate rises from 2\% below $x=.25$ to 100\% at $x\ge1$, a step at the predicted place. As a ranking predictor, $x$ reaches an area under the receiver operating characteristic curve (AUC) of .88--.99 on the 794-passage graph and .96--1.00 on the three larger ones, and a predicted flip almost always happens, as Appendix~\ref{app:law} details.

The prediction also matches the measured rate budget by budget. As Figure~\ref{fig:budget} shows, the predicted rate falls inside the measured 95\% interval on the 794-passage graph and within 8 points on the three larger graphs, and Table~\ref{tab:law} lists the values. We also find that the step does not depend on the models that build the graph. With Llama-3.1-8B extracting the same 794 passages, $x$ reaches AUC .95--.99 at each budget, and with Llama-3.3-70B and NV-Embed-v2, the models of HippoRAG2's default configuration, it reaches AUC .99--1.00 on the 1,000-question HotpotQA corpus.

\paragraph{Predictions fixed before the attack.}
To test whether $x$ predicts flips before they happen, we built a fifth HippoRAG2 graph from HotpotQA training questions drawn with a fixed random seed. Before any attack, we computed $x$ for the PageRank queries holding gold in the clean top-5 and recorded these predictions. As shown in Figure~\ref{fig:budget}(f), the predictions underestimate the measured rate at each budget by at most 5.3 points and rank flips at AUC .99. Table~\ref{tab:prereg} lists the exact values, and Appendix~\ref{app:law} traces the underestimate to the first-order term of the walk.

\paragraph{Entity-targeted writes.}
We further let the writer point its edges at an \emph{entity} rather than a record. The captured mass then spreads over all records at that entity and drains the records at the seeds, which is the case Equation~\eqref{eq:x} is written for. We run these writes at two weights on the conversation memory of Section~\ref{sec:rule}. Figure~\ref{fig:collapse}(b) shows the step in the same place, with the flip rate below 13\% for $x<.25$ and above 85\% for $x\ge2$. The full ratio ranks these queries at AUC .87 and .90, whereas capture alone reaches only .62--.72, which shows that the drain term carries the prediction here.

\paragraph{Text-only writes.}
We find that text-only writes also flip selections, and the flip condition ranks which ones. We test a writer that adds one short passage per victim seed entity, each relating it to an entity of the first record outside the top-5, and let HippoRAG2's extractor turn the passages into edges. On 200 held-out HotpotQA queries, the passages make unseen changes on 36.5\% [30.0\%, 43.5\%] of the queries, or 33.5\% with HippoRAG2's default models, while passages of the same form about unrelated entities change nothing. They also bring the designated record into the top-5 on 41.5\% of the queries. Computed from the extracted edges, $x$ ranks the unseen changes at AUC .85 and the designated record's entry at .90.

\paragraph{Effect of the write budget.}
We find that beyond a few units, a higher weight barely raises the flip rate. $W$ enters $\Delta$ only through $\theta_s=W/(W+\str(s))\in[0,1)$. Capture rises and saturates in $W$, at a limit $(1-\alpha)\sum_s\pi(s)$ fixed by the clean graph.

\begin{proposition}[Saturation]
\label{prop:sat}
To zeroth order, the set of queries a writer can ever flip by record-targeted edges from the seeds to a record $w$ is $E_\infty=\{Q:\ m_k-\pi(w)<(1-\alpha)\sum_s\pi(s)\}$, a property of the clean graph. At unit weight the writer already holds $1/(1+\str(s))$ of the limit at each seed.
\end{proposition}

\noindent\emph{Proof.} $\kappa(W)$ increases to $(1-\alpha)\sum_s\pi(s)$, so the flip condition $\kappa(W)>m_k-\pi(w)$ holds at some $W$ exactly on $E_\infty$. \qed

We observe that the clean graph predicts where the flip rate saturates. As the dotted lines in Figure~\ref{fig:budget} show, the set $E_\infty$ predicts the saturated flip rate on all four graphs to within 4 points, and both rates saturate by $W\approx 16$. Seed entities in these graphs have small strength, so six edges of weight at most one already flip 32--50\% of queries.

\subsection{Flip Condition for Leiden}

The same comparison governs community selectors. Whereas PageRank ranks records by walk mass, community selectors answer a query from the summaries of its seeds' communities. Microsoft GraphRAG builds those communities with Leiden~\cite{leiden2019}, which partitions the graph by modularity. Under Leiden a writer changes modularity gain rather than moving walk mass. The computation is still global and the structure still writable, and a condition of the same form holds, with the two numbers written in modularity.

Let $m$ be the total edge weight, $k_v$ the strength of entity $v$ in community $A$, $K_B$ the strength of community $B$, $k_{vB}$ the weight from $v$ into $B$, and $e_{AB}$ the weight between $A$ and $B$. The writer adds $W$ unit edges between $v$ and $B$.

\begin{samepage}
\begin{proposition}[Leiden thresholds]
\label{prop:leiden}
The modularity gain of moving $v$ into $B$, and of merging $A$ into $B$, after the write are
\begin{align*}
\Delta Q_{\mathrm{node}}(W)&=\frac{k_{vB}+W-k_{vA'}}{m+W}\\
&\quad-\frac{(k_v+W)(K_B+W-K_{A'})}{2(m+W)^2},\\
\Delta Q_{\mathrm{merge}}(W)&=\frac{e_{AB}+W}{m+W}-\frac{(K_A+W)(K_B+W)}{2(m+W)^2},
\end{align*}
where $A'=A\setminus\{v\}$. Each is a rational function of $W$ that changes sign at most once on $W\ge0$. Let $W^{*}$ be the smallest number of edges at which either gain is positive. If the clean partition is a Leiden optimum and the writer's edges join $v$ to $B$ only, then one of the two moves that take $v$ into $B$, moving $v$ or merging $A$ into $B$, has positive gain if and only if $x=W/W^{*}\ge1$.
\end{proposition}
\end{samepage}

Each gain, multiplied by $(m+W)^2$, is a quadratic in $W$ that is not positive at $W=0$, as Appendix~\ref{app:proofs} shows. Intuitively, once the writer has added $W^{*}$ edges, moving the entity into $B$ raises modularity, and $W^{*}$ is read off the clean graph. We test whether Leiden makes that move on two entity graphs that Graphiti extracted: a 440-target graph built from Wikipedia-style passages and a 94-target graph built from a LoCoMo conversation. On each graph we write at eight weights and, as GraphRAG does, re-cluster the whole graph after each write with a fixed seed (4,272 rows).

We find that Leiden makes the move at the step, as Figure~\ref{fig:collapse}(c) shows. No target moves below $x=.25$, 99.9\% move for $x\in[1,2)$, and $x$ ranks the moves at AUC .98 or higher on both graphs, while the per-budget rates in Table~\ref{tab:law} match the prediction within 2 points on the larger graph. GraphRAG's own community builder puts the step in the same place on its default component, as Table~\ref{tab:census} and Appendix~\ref{app:law} show. The median $W^{*}$ is at most two edges, within reach of a few turns of plain talk, as Section~\ref{sec:census} shows.

\paragraph{Summary of the analysis.}
Both numbers come from the clean graph: $E_\infty$ and the boundary gap for PageRank, and $W^{*}$ for Leiden. An operator can therefore compute the queries a memory exposes before any write. Section~\ref{sec:rule} turns the same analysis into a defense that decides which edges count.

\begin{table*}[t]
\centering
\caption{\textbf{Clean utility and attack robustness of each defense.}}
\label{tab:defenses}\label{tab:fourway}
\begingroup\footnotesize\setlength{\tabcolsep}{3pt}\renewcommand{\arraystretch}{0.98}\begin{tabularx}{\textwidth}{l c r@{\hspace{3pt}}l *{6}{Y}}
\toprule
& \multicolumn{3}{c}{Clean queries} & \multicolumn{6}{c}{Gold hit@5\,$\uparrow$ under the write that targets} \\
\cmidrule(lr){2-4}\cmidrule(l){5-10}
 & & & & record & \multicolumn{2}{c}{one entity} & \multicolumn{2}{c}{several entities} & gold entity \\
\cmidrule(lr){6-7}\cmidrule(lr){8-9}
Defense & recall@5\,$\uparrow$ & \multicolumn{2}{c}{cost [95\% CI]\,$\downarrow$} & $W{=}120$ & $W{=}4$ & $W{=}120$ & $W{=}4$ & $W{=}120$ & $W{=}120$ \\
\midrule
\multicolumn{10}{c}{\textit{(a) Primary defenses, 694 questions (665 for entity-targeted writes)}} \\[0.5pt]
no defense & .395 & \multicolumn{2}{c}{---} & .501 & .585 & .546 & .586 & .544 & .472 \\
write authorization & .215 & .179 & [.157, .203] & .317 & .317 & .317 & .317 & .317 & .317 \\
remove all writer structure & .331 & .064 & [.039, .086] & .478 & .478 & .478 & .478 & .478 & .478 \\
provenance gate, $r{=}1$ & .210 & .185 &  & .542 & .585 & .546 & .586 & .544 & .472 \\
provenance gate, $r{=}2$ & .001 & .394 &  & .169 & .027 & .027 & .027 & .027 & .006 \\
provenance gate, $r{=}3$ & .001 & .394 &  & .003 & .003 & .003 & .003 & .003 & .003 \\
\rowcolor{p2rule}\textbf{own-records rule} & .396 & $-$.001 & [$-$.014, .011] & .586 & .588 & .588 & .588 & .588 & .588 \\
\midrule
\multicolumn{10}{c}{\textit{(b) Additional structural defenses, 665 questions with attacks rerun}} \\[0.5pt]
no defense & .395 & \multicolumn{2}{c}{---} & .505 & .585 & .550 & .589 & .549 & .484 \\
delete flagged edges (oracle) & .395 & 0 &  & .591 & .591 & .591 & .591 & .591 & .591 \\
share cap, $\tau{=}.25$ & .375 & .019 & [.001, .038] & .507 & .541 & .528 & .555 & .532 & .582 \\
TrustRank-style restart & .212 & .183 & [.159, .204] & .266 & .311 & .289 & .313 & .289 & .254 \\
\rowcolor{p2rule}\textbf{own-records rule} & .396 & $-$.001 & [$-$.014, .009] & .588 & .588 & .588 & .588 & .588 & .588 \\
\bottomrule
\end{tabularx}
\endgroup
\pTwoNotes{HippoRAG2 PPR on LoCoMo, 694 questions (665 for entity-targeted writes and in part (b)). Cost is the drop in clean recall@5 from no defense, with 95\% cluster-bootstrap intervals. Gate cost flags every edge the writer added, and under attack the gates flag the attack edges. Part (b) is a separate run on the same graphs.}
\end{table*}

\begin{figure*}[t]
\centering
\includegraphics{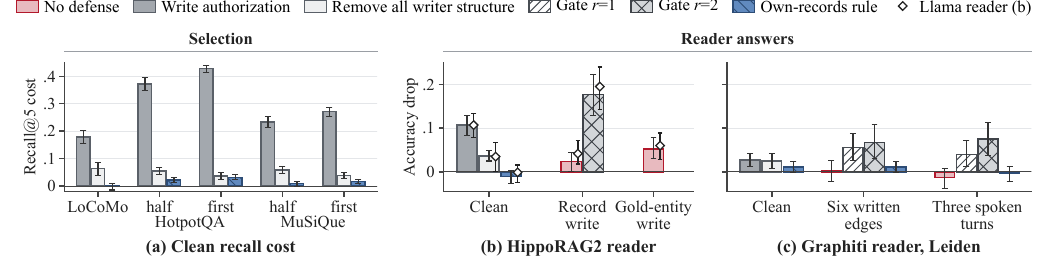}
\caption{\textbf{Defense costs on clean queries and in the reader's answers.} Whiskers are 95\% confidence intervals.}
\label{fig:cost}
\end{figure*}

\section{The Own-Records Rule}
\label{sec:rule}

We present the own-records rule and evaluate it as a defense (Q2). Proposition~\ref{prop:decomp} shows which edges carry the attack, and \textbf{the own-records rule removes exactly those edges while keeping all records and all edges that touch their writer's own records.}

\subsection{Design and Guarantees}

We derive the rule from where the attack enters the walk. In Proposition~\ref{prop:decomp}, the non-zero entries of $\Delta$ sit at the endpoints of the written edges and next to them: mass gained at the far end and lost at the neighbours of the near end. An edge can \emph{lie outside $U$'s own records}: it can run from a seed to a victim record, from a seed to a shared entity, or from $U$'s entity to another entity. Such an edge moves mass between the victim's records with no record of $U$'s on the way. An edge with one of $U$'s records at an end moves mass only into or out of that record, and any mass that moves on starts there. When the selector runs for party $V$, the rule keeps every record and every edge, but lets the edges written by any other party $U$ carry mass only when they are incident to $U$'s own records. Everything else $U$ wrote is skipped for this walk, and what is left of $U$'s influence passes through records that carry $U$'s identity.

PageRank takes two inputs, the transition matrix and the restart vector, and the rule says the same thing about both: what a writer adds counts only at its own records. For the transition matrix this is the edge filter above, while the restart vector takes two more pieces. First, the walk restarts at records in proportion to how well their text matches the query, so a matching record of $U$ could lend that score to its neighbours over edges the rule keeps. In $V$'s walk, the content score of another party's record therefore counts only on that record: a walk that restarts there ends there. Second, HippoRAG2 divides an entity's restart weight among the passages that name it, and passages of $U$ that name $V$'s seed entities would dilute it. Under the rule, that weight is divided only among passages of $V$, of parties $V$ trusts and of the system. Walk mass that reaches $U$'s records from $V$'s seeds still flows through them to whatever lies behind them. Section~\ref{sec:ruleaware} measures what each piece of the rule removes.

We contrast the rule with the two most direct defenses. \emph{Write authorization} refuses $U$'s writes, records and structure together, which costs $V$ all of $U$'s content. \emph{Removing all of $U$'s structure} keeps $U$'s records but cuts them off from the shared entities, leaving $V$ no path to them. The own-records rule keeps the content and keeps the path to it.

\paragraph{Implementation.}
We implement the rule as a filter on the transition matrix and on the restart, applied when the walk is built for a query from the writer identity that multi-writer memories already record for each edge. $V$ can treat a party it trusts, such as an agent that consolidates the shared memory, as itself. A change $U$ makes to another party's object, such as a new weight or a mark that a fact is invalid, is skipped the same way. The store keeps everything, and a system with one writer is untouched.

\paragraph{Edge attribution.}
The filter needs a writer for each edge. Derived structure takes the writer of its source: an extracted edge or fact belongs to the writer of the text it came from, a synonym edge to a writer whose records alone mention one of its entities, and an expiry to the speaker whose turn caused it. Edges that both parties' records support, and system edges, are kept in all walks. The selector itself runs unchanged on the filtered graph, while community selectors recompute their communities once per party, with summaries that keep all facts. Appendix~\ref{app:rule} gives the details for each system.

\paragraph{Runtime cost.}
The rule costs one filtered walk per query for PageRank, which takes 0.04--0.05 seconds or 1.9--3.2\% of end-to-end latency on HippoRAG2. For Leiden it costs one clustering per party, which takes under 0.01 seconds on the ten Graphiti graphs of Section~\ref{sec:reader}. The content part and the community summaries are computed once per party when the graph changes, outside the query path.

\paragraph{Security guarantees.}
Once the writer knows the rule, it has one way left to move the victim's walk: attach its structure to its own records, where the edges survive the filter. The mass those edges move then passes through the writer's records, and what the writer moves is bounded by what they receive.

\begin{proposition}[Influence bound under the rule]
\label{prop:bound}
Under the own-records rule, let $A$ be the records of writer $U$ that touch an edge $U$ added or reweighted, and let $\omega$ and $\omega'$ be the victim's walk before and after the write, counting only walks that restart outside $U$'s records. On the records $U$ does not own, the walk changes by at most $\|\omega'-\omega\|_1\le\frac{2}{\alpha}\big(\omega(A)+\omega'(A)\big)$. If the write also changes the restart weights outside $U$'s records from $r$ to $r'$, the right side gains $\|r'-r\|_1$.
\end{proposition}

The proof, given in full in Appendix~\ref{app:proofs}, couples the walks before and after the write until one of them enters $A$. Since $\omega(A)$ is at most the score of $A$, the rule enforces Definition~\ref{def:si} with $c=2/\alpha$ against every write that leaves the restart weights unchanged, which includes every write of structure. It specializes the PageRank stability bound of Ng, Zheng and Jordan~\cite{ng2001stable}, which counts all nodes whose out-links change, to the writer's own records.

\begin{corollary}[Condition for an unseen change]
\label{cor:flip}
Under the own-records rule, suppose $V$'s clean top-$k$ holds no record of $U$, and let $g$ be the gap between the $k$-th and $(k{+}1)$-th scores among the records $U$ does not own. A write of $U$ can make an unseen change only if $\omega(A)+\omega'(A)\ge\alpha g/2$.
\end{corollary}
\noindent\emph{Proof.} An unseen change swaps two records that $U$ does not own across the boundary. Their scores must then move by at least $g$ in total, which needs $\|\omega'-\omega\|_1\ge g$, and Proposition~\ref{prop:bound} bounds the left side by $\frac{2}{\alpha}\big(\omega(A)+\omega'(A)\big)$. \qed

\noindent Intuitively, the flip condition of Section~\ref{sec:law} still holds under the rule, with the score a writer's edges can add replaced by the mass its own records hold.

Conversely, no filter that judges edges one at a time can keep more than the rule and still have selection integrity. The reason is that an edge touching none of the writer's records needs no record of the writer to exist: it can move the walk while the right side of Definition~\ref{def:si} stays at zero.

\begin{proposition}[No larger filter]
\label{prop:max}
Let a filter decide how much of each edge's weight counts in $V$'s walk from the edge's writer and from which of its endpoints are that writer's records. If it lets an edge of a writer $U\neq V$ with no endpoint among $U$'s records count with positive weight, PageRank under the filter does not have selection integrity for $V$ against $U$. The rule's edge part lets every other edge of $U$ count, which makes it the largest such filter with selection integrity.
\end{proposition}

\noindent The proof is in Appendix~\ref{app:proofs}. As Section~\ref{sec:ruleaware} shows, the influence the rule still allows passes through the writer's records, and no rule that sees only the graph and its writers can remove it without cutting the bridges the victim needs.

\subsection{Comparison with Other Defenses}

To evaluate the rule against other defenses, we compare them on LoCoMo~\cite{locomo2024}, long two-person conversations whose questions need evidence from both speakers' sessions. Each conversation gives one HippoRAG2 graph, and each speaker serves once as victim and once as writer: 20 pairs and 694 questions, more than half of which need a turn of the other speaker. The six edge patterns target one wrong turn, one or more wrong entities at two weights each, or the gold entity.

Table~\ref{tab:defenses}(a) shows that only the rule keeps both the clean recall and the gold under attack. Its clean cost is $-0.1$ points of recall@5, with an interval containing zero, against 17.9 points for write authorization, and it stays the same with four or six parties. All six attacks write edges outside the writer's own records, the edges that Proposition~\ref{prop:decomp} identifies as carrying the attack, and the rule keeps the gold at its clean rate under all six. Write authorization and removing all writer structure resist the attacks too, but drop the writer's content or its links even on clean queries.

Moreover, the rule's cost stays small even when another party writes much of what the victim needs. Figure~\ref{fig:cost}(a) shows this on HotpotQA and MuSiQue graphs whose passages we split between two writers, either at random (half) or so that the other party writes the first passage of each multi-hop question (first). With a random split, the rule costs at most 2.2 points of recall@5, less than half of what removing all writer structure costs. When the other party writes the first passages, some of the victim's bridges run over that party's entity-to-entity edges, and the cost rises to at most 3.2 points, still below that of removing all writer structure, as Table~\ref{tab:split} shows.

We also evaluate two other structural defenses: a cap on the share of any entity's out-weight that one writer may hold, and a restart only at the victim's own nodes in the manner of TrustRank~\cite{trustrank2004}. Both fall short in Table~\ref{tab:defenses}(b): the cap weakens the attack without stopping it, and the restart leaves the attack edges in place and loses 18 points of clean recall. Table~\ref{tab:baselines} gives the full grid.

We observe that the attack columns of Table~\ref{tab:defenses}(a) differ as the flip condition predicts. The record-targeted write moves seed mass onto one wrong turn and changes nearly all selections, while the entity-targeted writes spread that mass over an entity and change four selections in ten by drawing it away from the seeds' neighbours. Edges to the gold entity push the gold out most often: a writer needs the victim's entities, not the answer. As Appendix~\ref{app:rule} shows, 27--35\% of the changed selections are unseen changes, and a record-level defense can reach them only from the structure near the selected records, as the gate does.

\subsection{Analysis of the Provenance Gate}

We explain why the provenance gate hurts the victim instead of the writer. The gate drops selected records near a flagged edge, and the writer chose both endpoints of that edge: a seed entity of the victim's query at one end, and its target at the other. At radius 1 the gate drops the target, which for the record-targeted write recovers half of the lost gold hits. Radius 1 changes no selection for the entity-targeted writes, which target entities rather than records.

At radius 2, however, the gate also drops the records that share an entity with a flagged endpoint. The seed is one such endpoint, and its records are dropped with the gold among them. Gold hits then fall below no defense under all six attacks, and to 0.3\% at radius 3, as shown in Table~\ref{tab:defenses}(a). On clean queries, where it flags all edges of the other writer, radius 2 leaves recall@5 at 0.1\%, because whatever structure it flags is near the records the victim needs.

Acting after selection, the gate can only drop records, and the writer decides which ones by choosing where its edges end. Deleting the flagged edges and recomputing instead acts before selection, as the own-records rule does, but needs an oracle naming the attack edges. Given one, it keeps the gold on 59.1\% of queries under all six attacks in Table~\ref{tab:defenses}(b), and the rule reaches 58.8\% knowing only who wrote each edge.

\subsection{Effect on Answer Accuracy}
\label{sec:reader}

To see whether these differences reach the agent, we let the reader answer from each selection, and the defenses rank the same in its answers as in the selection. On HippoRAG2, because the 72B reader of Section~\ref{sec:setting} sees only the five selected turns, its accuracy tracks recall@5, as the accuracy drops in Figure~\ref{fig:cost}(b) show, and Table~\ref{tab:reader} lists the exact values. The rule matches the clean reader, write authorization and removing all writer structure cost accuracy even on clean queries, and the gate at radius 2 costs 17.7 points of accuracy, more than seven times the 2.4 points that the attack itself costs. A Llama-3.3-70B reader, shown as diamonds in Figure~\ref{fig:cost}(b), gives the same order, as Appendix~\ref{app:rule} details.

\paragraph{Results on Graphiti.}
We repeat the comparison on a second selector, extraction pipeline and reader input, and find that the rule again costs least and the gate hurts more than the attack. Graphiti extracts facts from ten LoCoMo conversations, Leiden clusters the extracted entity graph, and the reader answers from the community summaries and up to twenty member facts. Six edges from the gold entity to a wrong community cut the share of queries shown the gold entity's community from 84.5\% to 51.3\%. The attack itself costs 0.3 points of accuracy and the gate at radius 2 costs 6.7, while the rule's cost has an interval that contains zero, as shown in Figure~\ref{fig:cost}(c) and Table~\ref{tab:graphiti}.

\subsection{Adaptive Attacks}
\label{sec:ruleaware}

We consider an adaptive writer that knows the rule and adapts its writes to it. The rule leaves such a writer two levers, both through its own records, and on our memory only one of them matters: reducing the walk mass that entities forward. Under the rule every edge a writer adds must touch one of its own records. When those records hang off the victim's seeds, the mass they take from a seed only returns to the seeds, never to a victim's record the writer chose.

\begin{samepage}
\begin{proposition}[Silencing]
\label{prop:silence}
Under the own-records rule, let a writer tie new records with no restart weight only to the seeds $S$ of the victim's query, each record by edges of weight at most one, with total weight $m_s$ at seed $s$. Each seed then passes a share $\theta_s=m_s/(m_s+\str(s))$ of its forward mass to these records. As every $m_s$ grows, the victim's walk on its own records converges to $\pi^{S}$, the walk on the clean graph in which the seeds forward no mass, and no record of the writer stays in the top-$k$. In this limit the writer pushes the gold out exactly on the queries $C=\{Q:$ the gold is in the top-$k$ of $\pi$ but not of $\pi^{S}\}$, a property of the clean graph.
\end{proposition}
\end{samepage}

\begin{proof}The writer's records have no restart weight and edges only to seeds, so a walk that enters them from a seed can only step back to a seed. Eliminating them leaves a walk on the other nodes in which seed $s$ forwards a share $1-\theta_s$ of its mass along its clean edges and the rest, one step later, to the seeds. The fixed point $\alpha(I-(1-\alpha)P^{\top})^{-1}r$ is continuous in these shares, so as every $\theta_s\to1$ it tends to the walk in which no seed forwards mass to any other node, which equals $\pi^{S}$ off the seeds. A record receives from seed $s$ at most a share $1/m_s$ of that seed's forward mass, so it holds at most $1/\min_s m_s$, which tends to 0.
\end{proof}

\suppressfloats[t]
\begin{figure}[t]
\centering
\includegraphics{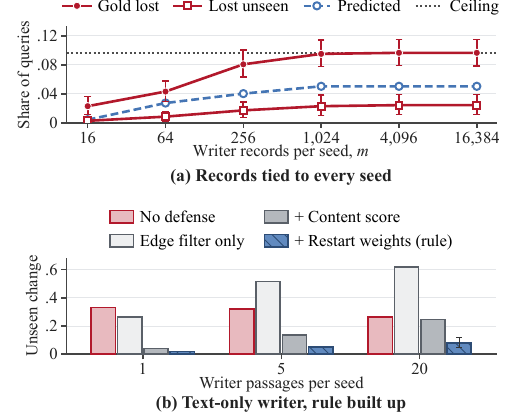}
\caption{\textbf{Adaptive writers under the own-records rule.} Panel (a) shows a writer that ties $m$ records to every seed, and panel (b) the text-only writer as the rule's pieces are added.}
\label{fig:ruleaware}\label{fig:silence}
\end{figure}

\noindent Intuitively, a writer who ties its records to the victim's seeds can at most silence them, and an operator can compute what that costs before anyone writes. As Figure~\ref{fig:ruleaware}(a) shows, the clean graph of our LoCoMo memory puts this cost at 9.7\% of queries in the limit, with the losses falling exactly where the proposition puts them. A writer that ties thousands of records to all seeds, more than the whole conversation holds, pushes the gold out on every predicted query and on no other. On three in four of these queries the victim's top-5 then also holds one of the writer's earlier turns, so the change is seen. The gold is lost unseen on only 2.4\% of queries. Across the 26 placements of Figure~\ref{fig:ruleaware}(a) and Appendix~\ref{app:rule}, every query pushed out unseen lies in the predicted set, and the zeroth-order estimate of Section~\ref{sec:law} ranks the losses at AUC .94. We further let the writer choose, for each query, which entities its records attach to: the seeds, the gold's entities or the entities of competing records. It searches greedily over at most 24 candidate entities per query and checks each placement on the exact walk of the victim. Even this optimized writer, with up to 1,024 records, pushes the gold out unseen on at most 9.5\% of queries, below the ceiling.

By contrast, the other lever, routing walk mass through the writer's records to an entity it chooses, adds nothing on our memory: each query that a routing writer pushed out unseen was also pushed out by a silencing writer. Synonym edges, which the system computes, open no route either, as Appendix~\ref{app:rule} details.

We find that a text-only writer also ends at silencing, because the rule covers the restart weights as well. Figure~\ref{fig:ruleaware}(b) shows how each piece of the rule removes part of its effect, and Table~\ref{tab:textwriter} lists the exact values. With edges filtered alone, its passages still lend their content score to their neighbours and dilute the seeds' restart weights, so with several passages per seed they make unseen changes more often than with no defense. The rule's content and restart pieces, which we designed after seeing this writer on 200 queries and measure on 200 others, remove the two effects. At the largest write in Table~\ref{tab:textwriter}, 8.0\% [4.5\%, 12.0\%] of queries still change unseen. Of these 16, 13 also change under full silencing, the limit of Proposition~\ref{prop:silence}, and none loses the gold.

\begin{table*}[t]
\centering
\caption{\textbf{Predicted and measured exposure of ten deployed selectors.}}
\label{tab:census}
\begingroup\footnotesize\setlength{\tabcolsep}{3pt}\renewcommand{\arraystretch}{0.98}\begin{tabularx}{\textwidth}{>{\bfseries}l l c c >{\raggedright\arraybackslash}X}
\toprule
\multicolumn{1}{l}{System} & Selector & Grows & Predicted & Measured \\
\midrule
HippoRAG2 & Personalized PageRank & yes & \pexposed & condition, Figure~\ref{fig:collapse}(a,b) \\
 & dense fallback (no seed) & no & \pclosed & 25--36\% of queries, unchanged \\
\midrule
Graphiti & node-distance rerank & no & \pclosed & user-centered: 0/430 selections changed, 0/135 with source-less edges \\
 & reciprocal-rank fusion & no & \pclosed & 0/135 selections changed \\
 & episode-mention rerank & yes & \pexposed & gold out of the top-10 on .83, reader harm .70 vs .02 (Section~\ref{sec:harm}) \\
\midrule
Mem0 & vector top-$k$ & no & \pclosed & selection unchanged \\
 & graph relations, BM25 & yes & \pexposed & gold out of the pool on .50 at the cap \\
\midrule
MS GraphRAG & communities (Leiden), global search & yes & \pexposed & Figure~\ref{fig:collapse}(c) condition, own builder on its default component moves .94 / .98 at $x\ge1$ and .00 / .04 below \\
 & local search, relationship count & yes & \pexposed & count condition: .999 / 0 / .361 \\
\midrule
LightRAG & local, entity-keyed chunks & no & \pclosed & 0/36 selections changed \\
\bottomrule
\end{tabularx}
\endgroup
\pTwoNotes{Grows: the structural term is unbounded in the writer's edges, or a budget rather than content cuts the pool.}
\end{table*}

Unlike a silencing writer, a writer with few records is seen, and writers designed to evade the filter are harmless or seen, as Appendix~\ref{app:rule} shows for each writer. Silencing, the lever that remains, cannot be removed by any filter that keeps the writer's content reachable.

\begin{corollary}[Silencing is the price of reachability]
\label{cor:price}
Proposition~\ref{prop:silence} holds under every filter that keeps the edges between $U$'s records and the entities they mention at full weight.
\end{corollary}

\noindent\emph{Proof.} The construction behind Proposition~\ref{prop:silence} uses only edges between $U$'s new records and the seeds they mention, with weights $m_s$ that grow. \qed

\noindent Going below the ceiling therefore means cutting the writer's records off the seeds or capping their weight. As Table~\ref{tab:defenses}(a) shows, removing all writer structure costs 6.4 points of recall@5, whereas the rule costs none.

\section{Exposure of Deployed Systems}
\label{sec:census}

We classify the selectors of deployed systems, since only those a writer can reach need the rule (Q3). The answer follows from the same comparison as the flip condition: \textbf{a selector is exposed when its structural term can grow past the boundary gap, and closed when that term cannot exceed that gap and the candidate pool is fixed by content.} PageRank and Leiden differ, but they share a form: a global objective, a structural term a writer can add to, and a gap that term has to cross.

\begin{proposition}[Closed selectors]
\label{prop:closed}
Let a selector score each candidate as $f(\text{content})+g(\text{structure})$ with $0\le g\le g_{\max}$, over a candidate pool fixed by content alone. A write that changes only structure can change the top-$k$ only if $g_{\max}$ is at least the gap between the $k$-th and $(k{+}1)$-th largest $f$.
\end{proposition}

\noindent\emph{Proof.} A pool fixed by content admits no candidate through a write of structure. Within the pool every score lies between $f$ and $f+g_{\max}$. If $g_{\max}$ is below that gap, the top-$k$ by $f+g$ is therefore the top-$k$ by $f$ for every structure, before and after the write. \qed

\paragraph{Classification from code.}
Two questions place a selector in Table~\ref{tab:census}, and both are answered from its code and the clean graph before any attack is run: can structure or a budget admit a candidate to the pool, and can the structural term grow with the edges a writer adds? A yes to either makes the selector exposed, and the flip condition of Section~\ref{sec:law} says which of its queries are exposed. Otherwise it is closed, and Proposition~\ref{prop:closed} confines a write to the queries whose content gap the bounded term can reach.

As Table~\ref{tab:census} shows, five of the ten selectors are exposed. PageRank and Leiden obey the flip condition, and three others rank by a count or fill a capped pool, among them Graphiti's mention count that Section~\ref{sec:harm} runs end to end. In the next two measurements, the writer writes rows into the graph store directly. Appendix~\ref{app:census} gives the same measurements for Mem0, exposed at its cap, and for LightRAG, which is closed.

\paragraph{Graphiti's node-distance reranker.}
We find that this reranker, centered on the user as Zep centers its memory, relies on structure yet stays closed. It orders a pool that vector and BM25 search fixed first by whether each fact's source is one hop from the center node, and the term is then already at its bound: all facts in the pool are one hop away on 392 of 430 LoCoMo questions, and ten added edges change no top-10 among the 430.

\paragraph{GraphRAG's local search.}
GraphRAG's local search, by contrast, is exposed through a count. It collects the text units of the entities a query matches, ranks each entity's units by how many of its relationships cite them, and fills a token budget from the top. The structural term is thus an uncapped integer count, while the budget cuts the pool. We let a writer add $R$ relationship rows citing a wrong text unit with clean count $c_w$. The writer moves that unit into the budget exactly when $c_w+R$ exceeds the count $c_{\mathrm{bnd}}$ at the budget boundary. Across all cases we test, the inequality decides the outcome: the wrong unit enters on 99.9\% of the cases where it holds and on none where it fails, with ties falling either way, as Appendix~\ref{app:census} details. One added row is enough to bring a wrong text unit into the budget on 63\% of queries. On 58\% of them, the row itself stays outside the reader's context.

\begin{figure}[t]
\centering
\includegraphics{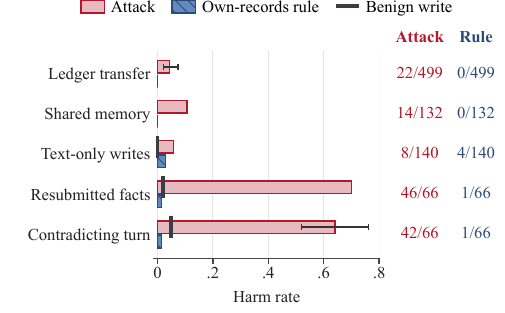}
\pTwoNotes{Counts are the cases the clean agent gets right: 66 answers on the 162 of 430 questions where ten resubmissions suffice, and 66 of 187 targets for contradicting turns.}
\caption{\textbf{Harm to the agent along five write paths, without defense and under the own-records rule.} Grey ticks mark benign writes.}
\label{fig:harm}\label{fig:harm_overview}
\end{figure}

\paragraph{Plain conversation turns.}
The writer need not write rows: plain turns are enough to move an entity through Graphiti's extraction. We let the other speaker in ten LoCoMo conversations add plain turns from three fixed templates linking a target entity to two or three high-degree members of the wrong community $B$. Graphiti ingests each turn on a copy of the graph, extracting a median of one target edge. In conversation 49, Sam's ``Apparently Prius and cabin go hand in hand now'' becomes a Prius--cabin fact. Prius, with clean threshold $W^{*}=1$, leaves its 60-entity community for the 3-entity one, and Evan's car question gets the wrong summary.

Proposition~\ref{prop:leiden} predicts which targets move, with $x$ the number of extracted edges at the target over $W^{*}$. As shown in Figure~\ref{fig:collapse}(d), the ratio ranks the moves at AUC .967, and no target below the step moves. Three turns move half of the targets into the wrong community and one turn a quarter, close to the predicted 57\% and 28\%.

The own-records rule returns all 60 targets to their clean communities, because the extracted facts belong to the writer's turns. As Table~\ref{tab:graphiti} shows, the gate again lowers reader accuracy and the rule does not.

All ten selectors behave as Table~\ref{tab:census} classifies them.

\section{Harm to Agents}
\label{sec:harm}

We measure what selection steering does to agents (Q4). \textbf{We find that writes the agent never reads misdirect both its actions and its answers}, whenever the reader sees the selected records themselves. Figure~\ref{fig:harm} summarizes five write paths, from a sandboxed action to a single spoken turn, each measured on the cases that the clean agent gets right. Under the own-records rule, the ledger and shared-memory paths lose no action, each of the four harms left after text-only writes shows the writer's passage to the reader, and resubmitted facts and contradicting turns, the two Graphiti paths, lose one answer each.

\paragraph{Ledger transfers and shared memory.}
We first give the agent a tool that appends a transfer to a sandboxed ledger, and pose questions whose answer names the recipient. On a 37,066-node HotpotQA graph over three seeds, a record-targeted write misdirects 22 of the 499 transfers the clean agent gets right, and a record-level check permits all of them. In the shared memory of Figure~\ref{fig:method}, where a compromised agent writes the edges, it misdirects 14 of 132. The rule prevents every misdirection in both settings.

\paragraph{Text-only writes.}
On the 200 held-out queries of Section~\ref{sec:law}, the text-only writer turns 8 of the reader's 140 correct answers wrong, 6 of them with none of its passages in the reader's context, and with HippoRAG2's default models it turns 6 of 159. Passages about unrelated entities turn none. Under the rule 4 remain in both settings, each with the writer's passage in the reader's context.

\paragraph{Resubmitted facts.}
We find that resubmitting facts the victim already has turns most correct answers wrong. Graphiti's mention-count reranker, also offered by Zep~\cite{zepsearch}, keeps the ten facts of a BM25 and vector pool that the most episodes mention. Its API for adding a fact, which Zep also exposes, merges an exact copy of an existing fact into that fact and counts one more mention. Table~\ref{tab:census} lists it as exposed, because the count grows with each write and the pool is cut at ten. On nine LoCoMo conversations, the other speaker resubmits a median of five of the victim's own facts, on the questions where ten resubmissions suffice. The gold then leaves the top-10 on 83\% of these questions, against 3\% for as many resubmissions from outside the pool. Of the questions the 72B reader answers correctly on the clean graph, it now gets 70\% wrong, against 2\% after benign writes, although no new fact reaches it. Counting a mention for the victim only when the victim or the fact's writer made it leaves one of these answers harmed at no clean cost, and Appendix~\ref{app:rule} shows that a Llama-3.3-70B reader agrees.

\paragraph{Contradicting turns.}
We find that one plain turn is enough when the memory expires facts, and that a published defense for poisoned memory does not stop it. When Graphiti ingests a turn that contradicts a fact, it marks the fact invalid, without recording whose turn caused the expiry. We attribute each expiry to the speaker of the turn whose ingestion caused it. On nine LoCoMo conversations one contradicting turn by the other speaker expires the victim's gold fact on 158 of 187 targets, and an unrelated turn expires none. An application that searches only valid facts then gets 64\% [52\%, 76\%] of its previously correct answers wrong, against 5\% after the unrelated turn, even with a record-level check removing all facts the other speaker supports. The facts the reader sees are all the victim's own or shared, leaving a defense for poisoned records nothing to catch. Figure~\ref{fig:expiry} traces one of these harms. As Appendix~\ref{app:harm} shows, the original A-MemGuard artifact~\cite{amemguard2025} prevents none of them, whereas under the rule with the same check the reader loses only one answer.

\paragraph{Relation to the flip condition.}
To see which questions are harmed, we sweep the record-targeted write over eight weights on the memory of Section~\ref{sec:rule} and let the reader of Table~\ref{tab:reader} answer each selection. Harm follows the step of Section~\ref{sec:law}. Of the 1,688 rows on questions the reader gets right on the clean graph, none of the 316 below the step gives a harmful answer, although the selection changes on 18\% of them, whereas the 1,372 above it are harmful at 8.1\% [4.1\%, 13.6\%].

\section{Related Work}
\label{sec:related}

\paragraph{Provenance for agent memory.}
Origin-bound authority~\cite{originbound2026} binds each memory record to its writer at write time. MAP-Graph~\cite{mapgraph2026} keeps a provenance graph over a shared multi-agent memory and gates retrieval and actions by ancestry and trust. SMSR~\cite{smsr2026} certifies retrieved records by keyed-hash provenance and ablation voting, and the non-amplification firewall~\cite{ppmf2026} stops summaries from laundering provenance. A-MemGuard~\cite{amemguard2025} validates retrieved memories by consensus, SuperLocalMemory~\cite{superlocalmemory2026} scores per-agent trust, and Collaborative Memory~\cite{collabmem2025} sets read and write policies over shared memory fragments. All of these decide about a record the agent reads. Section~\ref{sec:blind} shows that this decision is correct and that the attack lies outside the record. MemPoison~\cite{mempoison2026} finds that write-time defenses miss combinations of records, Utility Under Attack~\cite{utilityattack2026} finds provenance-weighted ranking indistinguishable from no defense on a flat memory, and Revoked but Still Authoritative~\cite{revoked2026} finds revocation unenforced at retrieval. We give these observations a mechanism: the selector moves conserved mass across writer boundaries, and the clean graph says how far.

\paragraph{Poisoning of retrieval and graph memory.}
PoisonedRAG~\cite{poisonedrag2025} and AgentPoison~\cite{agentpoison2024} inject text that a target query retrieves, and a single inserted blocker document~\cite{blocker2025} makes the model decline to answer the query that retrieves it. KEPo~\cite{kepo2026}, ShadowMerge~\cite{shadowmerge2026}, RAG Safety~\cite{ragsafety2025}, BadGraph~\cite{badgraph2026} and GraphRAG poisoning through source text~\cite{gragpoison2026,fewwords2025} turn injected text, interactions, triples or documents into graph knowledge, and LogicPoison~\cite{logicpoison2026} swaps entities in injected text to rewire the reasoning paths of GraphRAG and HippoRAG~2. These attacks put text, relations or triples in front of the reader, and a record-level check traces them to the writer. MemoryGraft~\cite{memgraft2025} and Oracle Poisoning~\cite{oraclepoisoning2026} write graph structure directly, as the writer of Section~\ref{sec:setting} does, and FragFuse~\cite{fragfuse2026} bypasses access control by fragmenting memory. We predict which selections such structure flips and confine it to its writer's records.

\paragraph{PageRank perturbation and graph adversaries.}
Link farms boost a page's PageRank with coordinated edges~\cite{linkspam2005,webspamtaxonomy2005}. TrustRank~\cite{trustrank2004} answers with a trust prior from a seed set. Sybil defenses~\cite{sybilguard2006,sybillimit2008,sybilrank2012} bound how far a walk from trusted nodes enters a region behind few attack edges. Attacks on graph neural networks and random-walk embeddings perturb structure to change a learned output~\cite{nettack2018,metattack2019,nodepoison2019}. Bojchevski and G\"unnemann certify how far a budget of edge changes can move a PageRank-based classifier~\cite{bojchevski2019certifiable}. Avrachenkov and Litvak~\cite{avrachenkov2006} analyze how new links change PageRank, and Hopcroft and Sheldon~\cite{hopcroft2007} design reputations that resist manipulation through hitting times. Building on the classical perturbation identity~\cite{ppr2003,localppr2006}, Proposition~\ref{prop:sat} shows that, for a writer whose edges start at the query's seeds, the level at which the flip rate saturates is a property of the clean graph. In PageRank an author controls only its own pages' out-links, so a link farm can redistribute only the walk mass its own pages receive. A graph memory drops that assumption, since any writer can attach edges at another party's seeds, and the own-records rule restores it. Unlike this line of work, we analyze the top-$k$ boundary of a memory with several writers. Our fix is a rule on whose edges the walk may use rather than a trust prior. Community deception hides a target from community detection by rewiring a few edges~\cite{waniek2018,fionda2018}, whereas our Leiden threshold reads the number of edges such a move needs off the clean graph.

\paragraph{Information-flow control.}
Lattice models~\cite{ifc1976}, decentralized labels~\cite{difc1997}, operating-system enforcement~\cite{histar2006} and taint tracking~\cite{taintdroid2010} track explicit flows into an output, and CaMeL~\cite{camel2025} and FIDES~\cite{fides2025} propagate such labels through the data an LLM agent handles and check them at tool calls. A tracker that followed the implicit flow of Section~\ref{sec:blind} would label each selected record with all writers in a shared memory. The own-records rule is an endorsement policy that stops this spread at its writer's records. Unlike Biba~\cite{biba1977} and Clark--Wilson~\cite{clarkwilson1987}, which keep low-integrity subjects from modifying high-integrity objects, it accepts all writes and limits per reader where another writer's structure counts.

\section{Discussion}
\label{sec:deploy}

\paragraph{Deployment.}
To deploy the rule, a memory needs one field per edge and per change to an existing object, such as an expiry: the writer of the text or turn it came from, which multi-writer memories already record for turns and consolidated facts. It leaves the store and the selector unchanged.

\paragraph{Limitations and future work.}
The two questions of Section~\ref{sec:census} classify any new selector and so tell an operator where to apply the rule. Proposition~\ref{prop:bound} proves the rule's bound for PageRank, and our experiments apply the rule to Leiden and mention counts. Proving the bound there is a natural next step, as is extending the flip condition to learned selectors such as graph neural network retrievers, with certified per-edge bounds~\cite{bojchevski2019certifiable} in the role of $\kappa$. Sybil-resistant identities~\cite{sybilguard2006,sybillimit2008} would extend the per-writer bound to a party with many identities, and capping the weight one writer may attach to a seed would also cap the share $\theta_s$ of Proposition~\ref{prop:silence} and thus the silencing it allows.

\paragraph{Disclosure.}
We have disclosed selection steering to the maintainers of the five systems we measured.

\section{Conclusion}
\label{sec:conclusion}

Graph memories with several writers let one writer steer which authentic records of another writer reach the agent. Record-level checks cannot see this selection steering, because the records they inspect are genuine. A flip condition derived for PageRank and Leiden predicts from the clean graph which selections a writer can flip. When fixed in advance, its predictions rank the measured flips at an AUC of .99. The own-records rule removes the edges that carry the attack, bounds what a writer can still change by the score of its own records, and costs no recall on a two-writer conversation memory. The same condition sorts deployed selectors into exposed and closed from their code, and on the exposed ones steering misdirects agents' actions and answers. Memories that agents, users and tools share therefore need a defense that protects how records are selected, not only who wrote them.

\bibliographystyle{IEEEtran}
\bibliography{refs}

\begin{thebibliography}{10}
\providecommand{\url}[1]{#1}
\csname url@samestyle\endcsname
\providecommand{\newblock}{\relax}
\providecommand{\bibinfo}[2]{#2}
\providecommand{\BIBentrySTDinterwordspacing}{\spaceskip=0pt\relax}
\providecommand{\BIBentryALTinterwordstretchfactor}{4}
\providecommand{\BIBentryALTinterwordspacing}{\spaceskip=\fontdimen2\font plus
\BIBentryALTinterwordstretchfactor\fontdimen3\font minus
  \fontdimen4\font\relax}
\providecommand{\BIBforeignlanguage}[2]{{%
\expandafter\ifx\csname l@#1\endcsname\relax
\typeout{** WARNING: IEEEtran.bst: No hyphenation pattern has been}%
\typeout{** loaded for the language `#1'. Using the pattern for}%
\typeout{** the default language instead.}%
\else
\language=\csname l@#1\endcsname
\fi
#2}}
\providecommand{\BIBdecl}{\relax}
\BIBdecl

\bibitem{hipporag2_2025}
B.~J. Guti\'errez, Y.~Shu, W.~Qi, S.~Zhou, and Y.~Su, ``From {RAG} to memory:
  Non-parametric continual learning for large language models,'' in \emph{Proc.
  ICML}, 2025.

\bibitem{graphrag2024}
D.~Edge, H.~Trinh, N.~Cheng, J.~Bradley, A.~Chao, A.~Mody, S.~Truitt,
  D.~Metropolitansky, R.~O. Ness, and J.~Larson, ``From local to global: A
  {Graph RAG} approach to query-focused summarization,'' 2024,
  arXiv:2404.16130.

\bibitem{zep2025}
P.~Rasmussen, P.~Paliychuk, T.~Beauvais, J.~Ryan, and D.~Chalef, ``{Zep}: A
  temporal knowledge graph architecture for agent memory,'' 2025,
  arXiv:2501.13956.

\bibitem{graphiti_repo}
\BIBentryALTinterwordspacing
{Zep}, ``{Graphiti}: Build real-time knowledge graphs for {AI} agents,'' 2026,
  {GitHub} repository, accessed 2026-06-03. [Online]. Available:
  \url{https://github.com/getzep/graphiti}
\BIBentrySTDinterwordspacing

\bibitem{poisonedrag2025}
W.~Zou, R.~Geng, B.~Wang, and J.~Jia, ``{PoisonedRAG}: Knowledge corruption
  attacks to retrieval-augmented generation of large language models,'' in
  \emph{Proc. USENIX Security}, 2025.

\bibitem{agentpoison2024}
Z.~Chen, Z.~Xiang, C.~Xiao, D.~Song, and B.~Li, ``{AgentPoison}: Red-teaming
  {LLM} agents via poisoning memory or knowledge bases,'' in \emph{Proc.
  NeurIPS}, 2024.

\bibitem{originbound2026}
Y.~Louck, ``Securing {LLM}-agent long-term memory against poisoning:
  Non-malleable, origin-bound authority with machine-checked guarantees,''
  arXiv:2606.24322, 2026.

\bibitem{mapgraph2026}
Y.~Wang, Z.~Yan, J.~Zhang, Z.~Wu, M.~Zheng, Z.~Sun, Y.~Zhu, and T.~Cai,
  ``{MAP-Graph}: Provenance-aware shared memory for multi-agent workflows,''
  arXiv:2608.10509, 2026.

\bibitem{smsr2026}
T.~Sharma, ``{SMSR}: Certified defence against runtime memory poisoning in
  persistent {LLM} agent systems,'' arXiv:2606.12703, 2026.

\bibitem{ppmf2026}
J.~Xu, Y.~Xiao, W.~Shao, H.~Liu, and X.~Li, ``Memory provenance laundering in
  {LLM} agents: A non-amplification firewall for persistent memory,''
  arXiv:2607.29167, 2026.

\bibitem{ppr2003}
G.~Jeh and J.~Widom, ``Scaling personalized web search,'' in \emph{Proc. WWW},
  2003.

\bibitem{lightrag2024}
Z.~Guo, L.~Xia, Y.~Yu, T.~Ao, and C.~Huang, ``{LightRAG}: Simple and fast
  retrieval-augmented generation,'' 2024, arXiv:2410.05779.

\bibitem{leiden2019}
V.~A. Traag, L.~Waltman, and N.~J. van Eck, ``From {Louvain} to {Leiden}:
  Guaranteeing well-connected communities,'' \emph{Sci. Rep.}, vol.~9, no.~1,
  p. 5233, 2019.

\bibitem{mem0_2025}
P.~Chhikara, D.~Khant, S.~Aryan, T.~Singh, and D.~Yadav, ``{Mem0}: Building
  production-ready {AI} agents with scalable long-term memory,'' 2025,
  arXiv:2504.19413.

\bibitem{zepfacttriple}
\BIBentryALTinterwordspacing
{Zep}, ``Manually updating the graph,'' 2026, {Zep} documentation, accessed
  2026-11-10. [Online]. Available:
  \url{https://help.getzep.com/adding-fact-triplets}
\BIBentrySTDinterwordspacing

\bibitem{mcpmemory}
\BIBentryALTinterwordspacing
{Model Context Protocol}, ``Knowledge graph memory server,'' 2026, reference
  server \texttt{@modelcontextprotocol/server-memory}, accessed 2026-11-10.
  [Online]. Available:
  \url{https://github.com/modelcontextprotocol/servers/tree/main/src/memory}
\BIBentrySTDinterwordspacing

\bibitem{contriever2022}
G.~Izacard, M.~Caron, L.~Hosseini, S.~Riedel, P.~Bojanowski, A.~Joulin, and
  E.~Grave, ``Unsupervised dense information retrieval with contrastive
  learning,'' 2021, arXiv:2112.09118.

\bibitem{localppr2006}
R.~Andersen, F.~Chung, and K.~Lang, ``Local graph partitioning using {PageRank}
  vectors,'' in \emph{Proc. IEEE FOCS}, 2006, pp. 475--486.

\bibitem{hotpotqa2018}
Z.~Yang, P.~Qi, S.~Zhang, Y.~Bengio, W.~W. Cohen, R.~Salakhutdinov, and C.~D.
  Manning, ``{HotpotQA}: A dataset for diverse, explainable multi-hop question
  answering,'' in \emph{Proc. EMNLP}, 2018.

\bibitem{twowiki2020}
X.~Ho, A.-K. {Duong Nguyen}, S.~Sugawara, and A.~Aizawa, ``Constructing a
  multi-hop {QA} dataset for comprehensive evaluation of reasoning steps,'' in
  \emph{Proc. COLING}, 2020.

\bibitem{musique2022}
H.~Trivedi, N.~Balasubramanian, T.~Khot, and A.~Sabharwal, ``{MuSiQue}:
  Multihop questions via single-hop question composition,'' \emph{Trans. Assoc.
  Comput. Linguist.}, vol.~10, pp. 539--554, 2022.

\bibitem{ng2001stable}
A.~Y. Ng, A.~X. Zheng, and M.~I. Jordan, ``Stable algorithms for link
  analysis,'' in \emph{Proc. ACM SIGIR}, 2001, pp. 258--266.

\bibitem{locomo2024}
A.~Maharana, D.-H. Lee, S.~Tulyakov, M.~Bansal, F.~Barbieri, and Y.~Fang,
  ``Evaluating very long-term conversational memory of {LLM} agents,'' in
  \emph{Proc. ACL}, 2024.

\bibitem{trustrank2004}
Z.~Gy\"ongyi, H.~Garcia-Molina, and J.~O. Pedersen, ``Combating web spam with
  {TrustRank},'' in \emph{Proc. VLDB}, 2004, pp. 576--587.

\bibitem{zepsearch}
\BIBentryALTinterwordspacing
{Zep}, ``Searching the graph,'' 2026, {Zep} documentation, accessed 2026-11-10.
  [Online]. Available: \url{https://help.getzep.com/searching-the-graph}
\BIBentrySTDinterwordspacing

\bibitem{amemguard2025}
Q.~Wei, T.~Yang, Y.~Wang, X.~Li, L.~Li, Z.~Yin, Y.~Zhan, T.~Holz, Z.~Lin, and
  X.~Wang, ``{A-MemGuard}: A proactive defense framework for {LLM}-based agent
  memory,'' 2025, arXiv:2510.02373.

\bibitem{superlocalmemory2026}
V.~P. Bhardwaj, ``{SuperLocalMemory}: Privacy-preserving multi-agent memory
  with {Bayesian} trust defense against memory poisoning,'' 2026,
  arXiv:2603.02240.

\bibitem{collabmem2025}
A.~Rezazadeh, Z.~Li, A.~Lou, Y.~Zhao, W.~Wei, and Y.~Bao, ``Collaborative
  memory: Multi-user memory sharing in {LLM} agents with dynamic access
  control,'' arXiv:2505.18279, 2025.

\bibitem{mempoison2026}
J.~Gao, K.~Xia, Y.~Zhang, X.~Hong, M.~Lin, X.~Wei, W.~Li, and S.~Lu,
  ``{MemPoison}: Uncovering persistent memory threats and structural blind
  spots in {LLM} agents,'' arXiv:2607.14651, 2026.

\bibitem{utilityattack2026}
A.~Karunanidhi, ``Utility under attack: Agent memory poisoning and the limits
  of content screening and provenance ranking,'' arXiv:2608.21230, 2026.

\bibitem{revoked2026}
Y.~T. Shen, K.~Toyoda, and A.~Leung, ``Revoked but still authoritative: An
  empirical study of revocation enforcement in agent-memory systems,''
  arXiv:2609.08258, 2026.

\bibitem{blocker2025}
A.~Shafran, R.~Schuster, and V.~Shmatikov, ``Machine against the {RAG}: Jamming
  retrieval-augmented generation with blocker documents,'' in \emph{Proc.
  USENIX Security}, 2025, pp. 3787--3806.

\bibitem{kepo2026}
Q.~Chen, C.~Qi, Y.~Huang, M.~Li, R.~Wang, D.~Zhang, K.~Qin, and S.~Liang,
  ``{KEPo}: Knowledge evolution poison on graph-based retrieval-augmented
  generation,'' The ACM Web Conference (WWW), arXiv:2603.11501, 2026.

\bibitem{shadowmerge2026}
Y.~Luo, Z.~Kang, T.~Ji, X.~Liu, Y.~Liu, S.~Li, and L.~Peng, ``{ShadowMerge}: A
  novel poisoning attack on graph-based agent memory via relation-channel
  conflicts,'' arXiv:2605.09033, 2026.

\bibitem{ragsafety2025}
T.~Zhao, J.~Chen, Y.~Ru, H.~Zhu, N.~Hu, J.~Liu, and Q.~Lin, ``{RAG} safety:
  Exploring knowledge poisoning attacks to retrieval-augmented generation,''
  arXiv:2507.08862, 2025.

\bibitem{badgraph2026}
L.~Yan, X.~Xu, and Z.~Li, ``{BadGraph}: Structural knowledge isolation attacks
  against graph retrieval-augmented generation,'' in \emph{Proc. USENIX
  Security}, 2026.

\bibitem{gragpoison2026}
J.~Liang, Y.~Wang, C.~Li, R.~Zhu, T.~Jiang, N.~Gong, and T.~Wang, ``{GraphRAG}
  under fire,'' in \emph{Proc. IEEE S\&P}, 2026.

\bibitem{fewwords2025}
J.~Wen, T.~Chen, Z.~Zheng, and C.~Huang, ``A few words can distort graphs:
  Knowledge poisoning attacks on graph-based retrieval-augmented generation of
  large language models,'' 2025, arXiv:2508.04276.

\bibitem{logicpoison2026}
Y.~Xiao, J.~Chen, Q.~Zhang, Y.~Zhang, C.~Zhou, L.~Yang, L.~Ren, X.~Yang, and
  X.~Huang, ``{LogicPoison}: Logical attacks on graph retrieval-augmented
  generation,'' in \emph{Proc. ACL}, 2026, pp. 5575--5591.

\bibitem{memgraft2025}
S.~S. Srivastava and H.~He, ``{MemoryGraft}: Persistent compromise of {LLM}
  agents via poisoned experience retrieval,'' 2025, arXiv:2512.16962.

\bibitem{oraclepoisoning2026}
B.~Kereopa-Yorke, G.~Diaz, H.~Wright, R.~Johnston, R.~F. Del~Rosario, and
  T.~Lynar, ``Oracle poisoning: Corrupting knowledge graphs to weaponise {AI}
  agent reasoning,'' 2026, arXiv:2605.09822.

\bibitem{fragfuse2026}
Z.~Rao, W.~Zhu, C.~A. Lu, Z.~Chen, W.~Niu, L.~Guan, B.~Li, and Z.~Xiang,
  ``{FragFuse}: Bypassing access control of large language model agents via
  memory-based query fragmentation and fusion,'' arXiv:2606.15609, 2026.

\bibitem{linkspam2005}
Z.~Gy\"ongyi and H.~Garcia-Molina, ``Link spam alliances,'' in \emph{Proc.
  VLDB}, 2005, pp. 517--528.

\bibitem{webspamtaxonomy2005}
------, ``Web spam taxonomy,'' in \emph{Proc. AIRWeb}, 2005.

\bibitem{sybilguard2006}
H.~Yu, M.~Kaminsky, P.~B. Gibbons, and A.~Flaxman, ``{SybilGuard}: Defending
  against sybil attacks via social networks,'' in \emph{Proc. ACM SIGCOMM},
  2006.

\bibitem{sybillimit2008}
H.~Yu, P.~B. Gibbons, M.~Kaminsky, and F.~Xiao, ``{SybilLimit}: A near-optimal
  social network defense against sybil attacks,'' in \emph{Proc. IEEE S\&P},
  2008.

\bibitem{sybilrank2012}
Q.~Cao, M.~Sirivianos, X.~Yang, and T.~Pregueiro, ``Aiding the detection of
  fake accounts in large scale social online services,'' in \emph{Proc. USENIX
  NSDI}, 2012.

\bibitem{nettack2018}
D.~Z\"ugner, A.~Akbarnejad, and S.~G\"unnemann, ``Adversarial attacks on neural
  networks for graph data,'' in \emph{Proc. ACM KDD}, 2018, pp. 2847--2856.

\bibitem{metattack2019}
D.~Z\"ugner and S.~G\"unnemann, ``Adversarial attacks on graph neural networks
  via meta learning,'' in \emph{Proc. ICLR}, 2019.

\bibitem{nodepoison2019}
A.~Bojchevski and S.~G\"unnemann, ``Adversarial attacks on node embeddings via
  graph poisoning,'' in \emph{Proc. ICML}, 2019, pp. 695--704.

\bibitem{bojchevski2019certifiable}
A.~Bojchevski and S.~G{\"u}nnemann, ``Certifiable robustness to graph
  perturbations,'' in \emph{Proc. NeurIPS}, 2019.

\bibitem{avrachenkov2006}
K.~Avrachenkov and N.~Litvak, ``The effect of new links on {Google}
  {PageRank},'' \emph{Stochastic Models}, vol.~22, no.~2, pp. 319--331, 2006.

\bibitem{hopcroft2007}
J.~Hopcroft and D.~Sheldon, ``Manipulation-resistant reputations using hitting
  time,'' in \emph{Proc. WAW}, 2007, pp. 68--81.

\bibitem{waniek2018}
M.~Waniek, T.~P. Michalak, M.~J. Wooldridge, and T.~Rahwan, ``Hiding
  individuals and communities in a social network,'' \emph{Nat. Hum. Behav.},
  vol.~2, no.~2, pp. 139--147, 2018.

\bibitem{fionda2018}
V.~Fionda and G.~Pirr{\`o}, ``Community deception or: How to stop fearing
  community detection algorithms,'' \emph{IEEE Trans. Knowl. Data Eng.},
  vol.~30, no.~4, pp. 660--673, 2018.

\bibitem{ifc1976}
D.~E. Denning, ``A lattice model of secure information flow,'' \emph{Commun.
  ACM}, vol.~19, no.~5, pp. 236--243, 1976.

\bibitem{difc1997}
A.~C. Myers and B.~Liskov, ``A decentralized model for information flow
  control,'' in \emph{Proc. ACM SOSP}, 1997, pp. 129--142.

\bibitem{histar2006}
N.~Zeldovich, S.~Boyd-Wickizer, E.~Kohler, and D.~Mazi\`eres, ``Making
  information flow explicit in {HiStar},'' in \emph{Proc. USENIX OSDI}, 2006,
  pp. 263--278.

\bibitem{taintdroid2010}
W.~Enck, P.~Gilbert, B.-G. Chun, L.~P. Cox, J.~Jung, P.~D. McDaniel, and
  A.~Sheth, ``{TaintDroid}: An information-flow tracking system for realtime
  privacy monitoring on smartphones,'' in \emph{Proc. USENIX OSDI}, 2010, pp.
  393--407.

\bibitem{camel2025}
E.~Debenedetti, I.~Shumailov, T.~Fan, J.~Hayes, N.~Carlini, D.~Fabian, C.~Kern,
  C.~Shi, A.~Terzis, and F.~Tram{\`e}r, ``Defeating prompt injections by
  design,'' in \emph{Proc. IEEE SaTML}, 2026, pp. 587--618.

\bibitem{fides2025}
M.~Costa, B.~K{\"o}pf, A.~Kolluri, A.~Paverd, M.~Russinovich, A.~Salem,
  S.~Tople, L.~Wutschitz, and S.~Zanella-B{\'e}guelin, ``Securing {AI} agents
  with information-flow control,'' in \emph{Proc. NeurIPS}, 2026, to appear,
  arXiv:2505.23643.

\bibitem{biba1977}
K.~J. Biba, ``Integrity considerations for secure computer systems,'' MITRE
  Corporation, Tech. Rep. MTR-3153, 1977.

\bibitem{clarkwilson1987}
D.~D. Clark and D.~R. Wilson, ``A comparison of commercial and military
  computer security policies,'' in \emph{Proc. IEEE S\&P}, 1987, pp. 184--194.

\bibitem{modularity2004}
M.~E.~J. Newman and M.~Girvan, ``Finding and evaluating community structure in
  networks,'' \emph{Phys. Rev. E}, vol.~69, no.~2, p. 026113, 2004.

\bibitem{deberta2021}
P.~He, X.~Liu, J.~Gao, and W.~Chen, ``{DeBERTa}: Decoding-enhanced {BERT} with
  disentangled attention,'' in \emph{Proc. ICLR}, 2021.

\end{thebibliography}

\appendices
\section{Detailed Proofs}
\label{app:proofs}

\begin{proof}[Proof of Proposition~\ref{prop:leiden}] The two expressions are the standard modularity gains of a move and a merge~\cite{modularity2004} after the $W$ new unit edges raise $m$, $k_v$, $K_A$, $K_B$, $k_{vB}$ and $e_{AB}$ by $W$ each and leave $K_{A'}$ unchanged. Multiplied by $(m+W)^2$, each is a quadratic in $W$ with leading coefficient $\tfrac12$. The clean partition is a Leiden optimum, so both are at most 0 at $W=0$, hence each stays at most 0 up to its larger root and is positive beyond it. One of the two is therefore positive exactly when $W\ge W^*=\min(W^*_{\mathrm{node}},W^*_{\mathrm{merge}})$.
\end{proof}

\begin{proof}[Proof of Proposition~\ref{prop:bound}] By the content part, a walk that restarts at a record of $U$ ends there, so the records $U$ does not own receive only $\omega$. Read $\omega$ as the distribution of where a walk from the restart weights outside $U$'s records stops when it stops with probability $\alpha$ at every step. Under the rule every edge that $U$ changed has an end in $A$. So at a node outside $A$ the two graphs differ only in how much weight leads into $A$, and the step taken given that it stays outside $A$ has the same distribution in both. Couple the two walks to take the same steps until one of them enters $A$. They stop at the same node unless one does, so $\frac12\|\omega'-\omega\|_1$ is at most the sum of the two probabilities of entering $A$. A walk that enters $A$ stops there at once with probability $\alpha$, so each probability is at most $\omega(A)/\alpha$ or $\omega'(A)/\alpha$. If the restart vector changes, the walk on $G'$ from $r$ differs from the walk on $G'$ from $r'$ by at most $\|r'-r\|_1$, because $\alpha(I-(1-\alpha)P'^{\top})^{-1}$ maps distributions to distributions, and the coupling bounds the rest.
\end{proof}

\begin{proof}[Proof of Proposition~\ref{prop:max}] Let $U$ write one such edge, of counted weight $W>0$, from a seed $s$ of $V$'s query to a record $w$ of $V$ that $s$ is not adjacent to, and no record. Then $A$ is empty, and Definition~\ref{def:si} requires $\pi'(w)=\pi(w)$. The fixed-point equation gives $\pi'(w)\ge(1-\alpha)P'_{sw}\pi'(s)\ge\alpha(1-\alpha)\theta_s r(s)$ with $\theta_s=W/(W+\str(s))$, because $\pi'(s)\ge\alpha r(s)$. This exceeds $\pi(w)$ whenever $w$ holds less than $\alpha(1-\alpha)\theta_s r(s)$ of the clean walk, as records far from the seeds do. Proposition~\ref{prop:bound} gives the rest.
\end{proof}

\section{Additional Results on the Flip Condition}
\label{app:law}

\begin{table}[t]
\centering
\caption{\textbf{Predictions fixed before the attack.}}
\label{tab:prereg}
\begingroup\footnotesize\setlength{\tabcolsep}{3pt}\renewcommand{\arraystretch}{0.98}\begin{tabularx}{\columnwidth}{c Y r@{\hspace{3pt}}l Y Y}
\toprule
$W$ & predicted & \multicolumn{2}{c}{measured [95\% CI]} & AUC of $x$ & median $\Delta/\kappa$ \\
\midrule
1 & .342 & .355 & [.322, .391] & .994 & 1.04 \\
2 & .482 & .525 & [.489, .560] & .990 & 1.05 \\
4 & .613 & .638 & [.600, .672] & .994 & 1.07 \\
8 & .710 & .761 & [.730, .791] & .991 & 1.10 \\
16 & .794 & .847 & [.821, .873] & .993 & 1.13 \\
32 & .851 & .896 & [.873, .918] & .994 & 1.16 \\
60 & .892 & .926 & [.906, .945] & .995 & 1.18 \\
120 & .910 & .944 & [.926, .959] & .995 & 1.20 \\
\bottomrule
\end{tabularx}
\endgroup
\pTwoNotes{HippoRAG2 on 1,000 HotpotQA training questions with 732 PageRank targets. Predictions were hashed before the attack. Intervals are 95\% CIs, and $\Delta/\kappa$ is the median true change relative to the zeroth-order capture.}
\end{table}

\begin{table*}[t]
\centering
\caption{\textbf{Predicted and measured flip rates across graphs and budgets.}}
\label{tab:law}
\begingroup\footnotesize\setlength{\tabcolsep}{3pt}\renewcommand{\arraystretch}{0.98}\begin{tabularx}{\textwidth}{l*{8}{Y}}
\toprule
Graph (PageRank) & $W{=}1$ & $W{=}2$ & $W{=}4$ & $W{=}8$ & $W{=}16$ & $W{=}32$ & $W{=}60$ & $W{=}120$ \\
\midrule
HotpotQA (794 passages) & .47 / .50 & .62 / .66 & .72 / .83 & .86 / .88 & .89 / .92 & .94 / .95 & .94 / .95 & .95 / .95 \\
\quad measured, 95\% interval & [.39, .62] & [.55, .77] & [.72, .91] & [.80, .95] & [.84, .98] & [.89, 1.00] & [.89, 1.00] & [.89, 1.00] \\
HotpotQA (1,000 q) & .33 / .35 & .47 / .51 & .62 / .67 & .71 / .77 & .78 / .83 & .83 / .89 & .87 / .92 & .90 / .94 \\
2Wiki (1,000 q) & .45 / .48 & .61 / .64 & .73 / .77 & .80 / .86 & .88 / .93 & .92 / .97 & .95 / .98 & .96 / .98 \\
MuSiQue (1,000 q) & .30 / .32 & .40 / .43 & .52 / .55 & .58 / .64 & .66 / .71 & .72 / .77 & .75 / .83 & .80 / .88 \\
HotpotQA (794), Llama-3.1-8B extractor & .50 / .53 & .65 / .67 & .80 / .79 & .85 / .89 & .89 / .91 & .91 / .92 & .92 / .92 & .92 / .94 \\
HotpotQA (1,000 q), default models & .21 / .23 & .34 / .36 & .47 / .48 & .57 / .60 & .68 / .73 & .77 / .81 & .82 / .86 & .87 / .92 \\
\midrule
Graph (Leiden) & $W{=}1$ & $W{=}2$ & $W{=}3$ & $W{=}5$ & $W{=}8$ & $W{=}13$ & $W{=}21$ & $W{=}34$ \\
\midrule
Graphiti graph (440 targets) & .35 / .35 & .89 / .91 & .95 / .97 & .97 / .98 & 1.00 / 1.00 & 1.00 / 1.00 & 1.00 / 1.00 & 1.00 / 1.00 \\
Graphiti graph (LoCoMo, 94 targets) & .73 / .73 & .73 / .73 & .73 / .73 & .73 / .74 & .74 / .74 & .74 / .74 & .79 / .76 & .94 / .77 \\
\bottomrule
\end{tabularx}
\endgroup
\pTwoNotes{Cells are prediction / measurement. The default models of HippoRAG2 are Llama-3.3-70B and NV-Embed-v2. Leiden uses $x=W/W^{*}$ from Proposition~\ref{prop:leiden}.}
\end{table*}

\paragraph{PageRank, record-targeted writes.}
The 16,452 rows of Figure~\ref{fig:collapse}(a) are 2,059 queries at eight budgets, less four budgets that one MuSiQue query did not complete and 16 rows whose target record already reaches the boundary, so that $x$ is undefined. The predicted $|E_\infty|$ is .97, .92, .97 and .84 of PPR queries on the four graphs, against measured flip rates at $W=120$ of .95, .94, .98 and .88. Misses at $W=120$ are 3--7\%, and false positives are 0 of 736 on HotpotQA, 0 of 603 on MuSiQue and 1 of 656 on 2Wiki. On the LoCoMo memory of Section~\ref{sec:rule}, the one-line form of the flip condition for the record-targeted write gives 675 true positives, 4 false positives and 15 misses over 694 queries.

\paragraph{Leiden.}
On the 440-target graph the node-only threshold $W^*_{\mathrm{node}}$ reaches AUC .961 with 318 misses, almost all cases in which the writer's edges merged $v$'s whole community into $B$ before $v$ alone would have moved. Taking the earlier of the two events gives .998 with 22 misses and 1 false positive. On the 94-target conversation graph the figures are .985 and .987, with 20 false positives and 2 misses.

\paragraph{Predictions fixed before the attack.}
The graph of Section~\ref{sec:law} built from 1,000 HotpotQA training questions covers their 9,793 passages and has 104,480 nodes and 556,149 edges. Of its 1,000 queries, 252 ground to no entity and fall back to dense retrieval.

\paragraph{First-order correction.}
We trace the underestimate of the fixed predictions to the first-order term, which also fixes where the step sits. The new edge gives the walk a way back to $w$: the true change at $w$ is $\kappa(1+\rho'_w)$, where $\rho'_w$ is the discounted return mass of $w$ on $G'$. The zeroth-order $x$ leaves out this return, which adds 4\% to 20\% (Table~\ref{tab:prereg}) and places the step just before $x=1$, and 64\% to 81\% of the flips the prediction missed have $x\in(.75,1)$.

\paragraph{GraphRAG's community builder.}
GraphRAG 2.7.2's \texttt{cluster\_graph} is graspologic's hierarchical Leiden with maximum cluster size 10 and GraphRAG's default seed (Table~\ref{tab:graphrag_leiden}). On the larger graph the largest connected component, which GraphRAG clusters by default, holds 58 of the 440 targets, and there the step stays sharp. On the full graph the move rate by bin of $x$ is 0 below .25, .40 in $[.75,1)$, .57 in $[1,2)$ and .90 in $[2,4)$.

\begin{table}[t]
\centering
\caption{\textbf{GraphRAG threshold validation.}}
\label{tab:graphrag_leiden}
\begingroup\footnotesize\setlength{\tabcolsep}{3pt}\renewcommand{\arraystretch}{0.98}\begin{tabularx}{\columnwidth}{l l Y Y Y Y}
\toprule
Graph & clustered & cells & AUC & $x\ge1$ & $x<1$ \\
\midrule
440-target & all nodes & 3,520 & .962 & .881 & .116 \\
 & largest component & 464 & .658 & .941 & .000 \\
\midrule
LoCoMo & all nodes & 752 & .984 & .965 & .040 \\
 & largest component & 736 & .987 & .978 & .040 \\
\bottomrule
\end{tabularx}
\endgroup
\pTwoNotes{GraphRAG 2.7.2 hierarchical Leiden at level 0, with eight budgets per target and re-clustering after each write. The last columns are move rates.}
\end{table}

\section{Additional Results on the Own-Records Rule}
\label{app:rule}

\begin{table}[t]
\centering
\caption{\textbf{Reader accuracy under defenses and writes.}}
\label{tab:reader}
\begingroup\footnotesize\setlength{\tabcolsep}{3pt}\renewcommand{\arraystretch}{0.98}\begin{tabularx}{\columnwidth}{>{\raggedright\arraybackslash}X c r@{\hspace{3pt}}l}
\toprule
Condition & accuracy & \multicolumn{2}{c}{drop from clean} \\
\midrule
clean & .303 & \multicolumn{2}{c}{---} \\
write authorization & .196 & .107 & [.083, .130] \\
remove all writer structure & .265 & .037 & [.025, .050] \\
\rowcolor{p2rule}\textbf{own-records rule} & .313 & $-$.010 & [$-$.026, .004] \\
\midrule
record-targeted write & .278 & .024 & [.001, .045] \\
record-targeted write, gate $r{=}2$ & .125 & .177 & [.128, .223] \\
gold-entity write & .248 & .053 & [.030, .078] \\
\bottomrule
\end{tabularx}
\endgroup
\pTwoNotes{HippoRAG2 on LoCoMo with the Qwen2.5-72B reader and 95\% cluster-bootstrap CIs. The gold-entity row covers 665 questions and uses the clean reader on that subset as its reference.}
\end{table}

\begin{table}[t]
\centering
\caption{\textbf{Rule ablation, text-only writer.}}
\label{tab:textwriter}
\begingroup\footnotesize\setlength{\tabcolsep}{3pt}\renewcommand{\arraystretch}{0.98}\begin{tabularx}{\columnwidth}{>{\raggedright\arraybackslash}X c c c}
\toprule
 & \multicolumn{3}{c}{Passages per seed} \\
\cmidrule(l){2-4}
Rule pieces & 1 & 5 & 20 \\
\midrule
no defense & .330 & .320 & .265 \\
edges counted only at own records & .265 & .515 & .620 \\
\quad + content score kept on its record & .040 & .135 & .245 \\
\rowcolor{p2rule}\textbf{\quad + restart weights from own passages} & .020 & .050 & .080 \\
\bottomrule
\end{tabularx}
\endgroup
\pTwoNotes{Unseen selection change on 200 held-out HotpotQA queries, as plotted in Figure~\ref{fig:ruleaware}(b). Writer passages are selected on .06--.16 with no defense, .24--.39 with edges filtered, and .38--.44 under the full rule. At 20 passages the rule's interval is [.045, .120].}
\end{table}

\begin{table}[t]
\centering
\caption{\textbf{Clean recall cost under shared authorship.}}
\label{tab:split}\label{tab:generalization}
\begingroup\footnotesize\setlength{\tabcolsep}{3pt}\renewcommand{\arraystretch}{0.98}\begin{tabular}{l l r@{\hspace{2pt}}l r@{\hspace{2pt}}l r@{\hspace{2pt}}l}
\toprule
 & & \multicolumn{2}{c}{write} & \multicolumn{2}{c}{remove all} & \multicolumn{2}{c}{own-records} \\
Graph & Split & \multicolumn{2}{c}{authorization} & \multicolumn{2}{c}{writer structure} & \multicolumn{2}{c}{rule} \\
\midrule
HotpotQA & half & .371 & [.347, .396] & .056 & [.043, .069] & .022 & [.013, .030] \\
 & first & .428 & [.415, .440] & .038 & [.026, .049] & .032 & [.023, .041] \\
\midrule
MuSiQue & half & .234 & [.215, .254] & .059 & [.046, .072] & .009 & [.000, .017] \\
 & first & .271 & [.255, .287] & .039 & [.028, .049] & .017 & [.008, .025] \\
\bottomrule
\end{tabular}
\endgroup
\pTwoNotes{Clean recall@5 cost against no defense on HippoRAG2, with 1,000 questions per graph and bootstrap 95\% CIs. Split: the other party writes half of the passages, at random (half), or the first passage of every multi-hop question (first). Shown as Figure~\ref{fig:cost}(a).}
\end{table}

\begin{table*}[!t]
\setlength{\abovecaptionskip}{2pt}
\centering
\caption{\textbf{Additional defenses and writer counts: full grid.}}
\label{tab:baselines}
\begingroup\footnotesize\setlength{\tabcolsep}{3pt}\renewcommand{\arraystretch}{0.98}\begin{tabularx}{\textwidth}{l c r@{\hspace{3pt}}l *{6}{Y}}
\toprule
& \multicolumn{3}{c}{Clean queries} & \multicolumn{6}{c}{Gold hit@5\,$\uparrow$ under the write that targets} \\
\cmidrule(lr){2-4}\cmidrule(l){5-10}
 & & & & record & \multicolumn{2}{c}{one entity} & \multicolumn{2}{c}{several entities} & gold entity \\
\cmidrule(lr){6-7}\cmidrule(lr){8-9}
Defense & recall@5\,$\uparrow$ & \multicolumn{2}{c}{cost [95\% CI]\,$\downarrow$} & $W{=}120$ & $W{=}4$ & $W{=}120$ & $W{=}4$ & $W{=}120$ & $W{=}120$ \\
\midrule
no defense & .395 & \multicolumn{2}{c}{---} & .505 & .585 & .550 & .589 & .549 & .484 \\
delete flagged edges (oracle) & .395 & 0 &  & .591 & .591 & .591 & .591 & .591 & .591 \\
share cap, $\tau{=}.1$ & .354 & .041 & [.028, .053] & .498 & .525 & .513 & .537 & .522 & .598 \\
share cap, $\tau{=}.25$ & .375 & .019 & [.001, .038] & .507 & .541 & .528 & .555 & .532 & .582 \\
share cap, $\tau{=}.5$ & .375 & .020 & [.000, .043] & .505 & .559 & .531 & .562 & .535 & .577 \\
TrustRank-style restart & .212 & .183 & [.159, .204] & .266 & .311 & .289 & .313 & .289 & .254 \\
\rowcolor{p2rule}\textbf{own-records rule} & .396 & $-$.001 & [$-$.014, .009] & .588 & .588 & .588 & .588 & .588 & .588 \\
\quad four parties & .396 & $-$.001 & [$-$.014, .009] & --- & --- & --- & --- & --- & --- \\
\quad six parties & .396 & $-$.001 & [$-$.014, .009] & --- & --- & --- & --- & --- & --- \\
\bottomrule
\end{tabularx}
\endgroup
\pTwoNotes{665 LoCoMo queries with 95\% cluster-bootstrap CIs. Attack columns match Table~\ref{tab:defenses}(a), and the rebuilt no-defense values differ by at most .012.}
\par\vspace{8pt}
\centering
\caption{\textbf{Rule-aware writes with few and many records.}}
\label{tab:ruleaware}
\begingroup\footnotesize\setlength{\tabcolsep}{3pt}\renewcommand{\arraystretch}{0.98}\begin{tabular}{l c r@{\hspace{3pt}}l c c c c c}
\toprule
 & & & & & selection & record & & share of bound \\
Variant & $m$ & \multicolumn{2}{c}{unseen gold loss [95\% CI]} & gold loss & changed & visible & gold hit@5 & (median / max) \\
\midrule
drain & 5 & .000 & [.000, .000] & .020 & .486 & .432 & .568 & .047 / .094 \\
 & 20 & .000 & [.000, .000] & .061 & .527 & .418 & .527 & .048 / .092 \\
 & 80 & .003 & [.000, .007] & .016 & .272 & .003 & .576 & .046 / .090 \\
\midrule
route & 5 & .000 & [.000, .000] & .021 & .496 & .435 & .568 & .056 / .103 \\
 & 20 & .000 & [.000, .000] & .060 & .529 & .411 & .531 & .053 / .097 \\
 & 80 & .003 & [.000, .007] & .015 & .278 & .002 & .579 & .049 / .094 \\
\bottomrule
\end{tabular}
\endgroup
\pTwoNotes{Drain ties each record to one seed, and route also ties it to a wrong-turn entity. Unseen loss counts only queries with no writer record selected (694 queries). Bound refers to Proposition~\ref{prop:bound}.}
\par\vspace{8pt}
\centering
\caption{\textbf{Graphiti reader utility across written and spoken attacks.}}
\label{tab:graphiti}
\begingroup\footnotesize\setlength{\tabcolsep}{3pt}\renewcommand{\arraystretch}{0.98}\begin{tabularx}{\textwidth}{l Y r@{\hspace{3pt}}l Y Y r@{\hspace{3pt}}l Y Y r@{\hspace{3pt}}l Y }
\toprule
 & \multicolumn{4}{c}{Clean graph} & \multicolumn{4}{c}{Six written edges} & \multicolumn{4}{c}{Three spoken turns} \\
\cmidrule(lr){2-5}\cmidrule(lr){6-9}\cmidrule(l){10-13}
Defense & accuracy & \multicolumn{2}{c}{drop [95\% CI]} & gold & accuracy & \multicolumn{2}{c}{drop [95\% CI]} & gold & accuracy & \multicolumn{2}{c}{drop [95\% CI]} & gold \\
\midrule
no defense & .107 & \multicolumn{2}{c}{---} & .845 & .104 & .003 & [$-$.021, .026] & .513 & .120 & $-$.013 & [$-$.038, .009] & .531 \\
write authorization & .080 & .027 & [.012, .043] & .627 & --- & \multicolumn{2}{c}{---} & --- & --- & \multicolumn{2}{c}{---} & --- \\
remove all writer structure & .082 & .025 & [.009, .042] & .710 & --- & \multicolumn{2}{c}{---} & --- & --- & \multicolumn{2}{c}{---} & --- \\
provenance gate, $r{=}1$ & --- & \multicolumn{2}{c}{---} & --- & .051 & .056 & [.027, .088] & .273 & .067 & .040 & [.012, .071] & .321 \\
provenance gate, $r{=}2$ & --- & \multicolumn{2}{c}{---} & --- & .040 & .067 & [.030, .109] & .078 & .032 & .075 & [.037, .114] & .203 \\
\rowcolor{p2rule}\textbf{own-records rule} & .096 & .011 & [$-$.001, .023] & .710 & .096 & .011 & [$-$.001, .023] & .710 & .110 & $-$.003 & [$-$.022, .013] & .710 \\
\bottomrule
\end{tabularx}
\endgroup
\pTwoNotes{748 LoCoMo questions, Leiden and the 72B reader. Drop is from the clean undefended reader, and gold is the share shown the gold entity's community.}
\end{table*}

\paragraph{Structural defenses and more writers.}
Table~\ref{tab:baselines} compares the rule with three structural defenses on the 665 LoCoMo questions where all six attacks apply. The TrustRank-style restart puts half of the restart mass on nodes the victim wrote, and for more parties each session of the writer becomes a separate writer whose structure the rule filters.

\paragraph{Adaptive attacks.}
With 5 or 20 records tied to one seed each, as in Table~\ref{tab:ruleaware}, the writer's records enter the top-5 on nearly nine in ten of the queries that lose the gold, and a record-level check that removes them returns the gold on all but 0.9\% of queries. We also let the writer write up to 1,280 records, tied to one seed or to every seed, in 20 placements. Unseen gold loss grows with the number of records: with one seed per record it is .003, .007 and .016 at 80, 320 and 1,280 records, and with every seed .007, .017 and .024. A writer that picks the best of the 20 placements for each query reaches .026 [.012, .041]. On 52 HotpotQA questions, the three writers designed to evade the filter with structure alone harm one question under the rule, and the one that ties its own record harms 5 of 25 rows, each time with that record in the reader's context. A writer can also coin a name close to a seed and to an entity it wants selected, so that synonym edges join them. Such a bridge formed on 29 of 200 targets and changed no selection unseen under the rule's attribution.

\paragraph{Fact expiry.}
In a separate run of the 187 contradicting turns of Section~\ref{sec:harm}, Graphiti expired the gold fact on 155 of them, and a search that keeps only valid facts lost the gold from its top-10 on .83 of these. On 153 of the 155, counting an expiry only when the victim or the fact's writer caused it keeps the gold valid and in the top-10.

\paragraph{Answer judging.}
On LoCoMo~\cite{locomo2024} an answer is correct when it contains the normalized gold string or a DeBERTa-v3 natural-language-inference model~\cite{deberta2021} scores its entailment above .5.

\paragraph{Second reader model.}
With Llama-3.3-70B as the reader of Table~\ref{tab:reader}, clean accuracy is .311, and the drops are $-$.001 [$-$.024, .016] for the rule, .035 for removing all writer structure, .107 for write authorization, .042 for the record-targeted attack, .195 for the gate at radius 2 and .060 for the gold-entity write. On the end-to-end attack of Section~\ref{sec:harm} it answers 67 of the 162 questions correctly, and .72 of these turn wrong under the attack, none under benign writes and one under the rule.

\paragraph{Unseen changes.}
Under Definition~\ref{def:blind}, the share of changed selections that are unseen changes is .335 for the record-targeted write, .351 and .273 for the write to one entity at $W=4$ and $W=120$, .338 and .275 for the write to several entities, and .278 for the gold-entity write.

\paragraph{Edge attribution.}
HippoRAG2's passage--entity and entity--entity edges belong to the writer of their passage, and a synonym edge computed from entity-name embeddings belongs to the system unless one writer's records alone mention one of its entities. A fact Graphiti extracts belongs to the speakers of its source episodes, and a GraphRAG or LightRAG relationship row to the writers of its source text units. Merging a mention into an existing entity adds no edge. A fact $U$ writes joins two entities without ending at one of $U$'s records, so it drops out of $V$'s community clustering while its text stays in the summaries.

\paragraph{Runtime cost.}
The filtered walk takes 0.038--0.052 s on the HippoRAG2 graphs of HotpotQA, 2Wiki and MuSiQue. On the ten Graphiti graphs of Table~\ref{tab:graphiti}, re-clustering for one party takes a median of 0.0047 s and at most 0.0073 s over the 20 victim--conversation pairs, against 0.0046 s unfiltered. In a party's partition, a median .83 of the communities have a member set absent from the unfiltered partition, and their summaries take about 46 calls to the 72B model per party.

\section{Additional Results on Deployed Systems}
\label{app:census}

\paragraph{Versions and write interfaces.}
We run HippoRAG2 at commit \texttt{d437bfb}, Graphiti 0.29.1 on Neo4j 5.26 and Mem0 1.0.11. Graphiti receives resubmitted facts through \texttt{add\_triplet} and turns through \texttt{add\_episode}, and its mention-count reranker is \texttt{episode\_mentions}.

\paragraph{GraphRAG's local search.}
The count condition $c_w+R>c_{\mathrm{bnd}}$ was evaluated over 1,116 (query, $R$) cells with Microsoft GraphRAG's own local-search ranking code, for the text units of one matched entity under budgets of 500, 700 and 1,000 tokens over HotpotQA passages of about 85 words. Ties, $c_w+R=c_{\mathrm{bnd}}$, admit the wrong text unit on .361 of cells and are broken by the order of the entity's text units.

\paragraph{Graphiti's reciprocal-rank fusion.}
The default recipe fuses vector and BM25 search over facts by reciprocal rank. Each of the nine conversations contributes 15 targets, each a fact with a source episode queried by its subject and relation. Besides the five source-less edges, the writer adds mentions that inflate the wrong entities' degree. Neither the top-10 sets nor their order changed on any target.

\paragraph{Mem0's graph memory.}
Mem0's graph search is exposed at its cap. The 72B model extracts relations from the first 200 turns of three LoCoMo conversations through Mem0's own \texttt{Memory.add}. Search matches query entities to nodes by embedding similarity of at least .7, collects their one-hop relations up to a cap of 100 and returns the BM25 top-5, so a writer's relation at a matched entity enters the pool by structure. Each of the 120 targets is a relation queried by its source and relation name, and the writer adds $W\in\{1,5,20,50,100,200\}$ relations from the source to new entities. On every target that stays below the cap, 43 of the 120 up to 20 relations and 42 at 50, no gold was pushed out unseen. Above the cap the gold leaves the pool on .013, .039, .091, .218 and .50 of queries at 1, 5, 20, 50 and 100 or more relations. One added relation changes the top-5 on .85 of queries, mostly by swapping BM25 ties.

\paragraph{LightRAG.}
LightRAG's chunk selection stays closed. It pools the source chunks of the entities a query matches. A source-less relation adds nothing to that pool, and 0 of 36 such writes change the selected chunks. A relation that does name a source chunk pulls the chunk into the context with the relation text, where a record-level check sees it.

\paragraph{Plain conversation turns.}
Each speaker of the ten conversations is once the victim, with three target entities per victim, 60 targets and 120 rows, re-clustered with the same Leiden call as Section~\ref{sec:law}. The clean-graph threshold $W^{*}$ has median 1.5. The move rates are 0/63 for $x<.5$, 0/6 for $x\in[.5,1)$, 21/25 for $[1,2)$ and 24/26 above 2. For the reader, the three-turn rows of Table~\ref{tab:graphiti} put the same three turns at all 748 questions of Section~\ref{sec:rule}, one ingestion per conversation, victim and gold entity (406 in all), with the reader, judge and clean rows unchanged. The gate is given every fact sourced from a writer turn. These facts touch the gold entity on .95 of rows, and the gate at radius 1 drops a median of four entities with all their facts.

\section{Additional Results on Harm to Agents}
\label{app:harm}

\paragraph{A-MemGuard.}
We ran the original A-MemGuard artifact (commit \texttt{52af44b}) with our 72B model on the 42 harmed contexts of Section~\ref{sec:harm} and their clean versions. Its memory auditor flags one context in each group, and its consensus check prevents none of the 42 harms and breaks 9 of the 42 clean answers by discarding the gold as an outlier.

\paragraph{Harm rates.}
Over three seeds (95\% CIs), the record-targeted write harms .031 [.015, .050] of 228 clean-correct HotpotQA rows, .198 [.124, .284] on MuSiQue, .044 [.020, .074] of the 499 ledger actions, and .08, .09, .11, .12, .13, .13, .18, .21, .23 and .40 on the ten LoCoMo conversations.

\end{document}